\documentclass[conference,letterpaper]{IEEEtran}
\usepackage{booktabs} 
\usepackage{amsmath,url,graphicx,amssymb}
\usepackage{amsfonts}
\usepackage{ctable}
\usepackage{multirow}
\usepackage{algorithm}
\usepackage{algpseudocode}
\usepackage{pifont}
\usepackage{color}
\usepackage{subfigure}
\usepackage{enumitem}
\usepackage{sidecap}
\usepackage{amsthm}

\newcommand{\mylistbegin}{
  \begin{list}{$\bullet$}
   {
     \setlength{\itemsep}{-2pt}
     \setlength{\leftmargin}{1em}
     \setlength{\labelwidth}{1em}
     \setlength{\labelsep}{0.5em} } }
\newcommand{\mylistend}{
   \end{list}  }

\newcommand{\eg}{\textit{e.g.}}

\newcommand{\ie}{\textit{i.e.}}
\newcommand{\etc}{\textit{etc}}

\newcommand{\wrt}{\textit{w.r.t.~}}

\newtheorem{theorem}{Theorem}[section]
\newtheorem{definition}[theorem]{Definition}
\newtheorem{corollary}[theorem]{Corollary}

\title{cube2net: Efficient Query-Specific Network Construction with Data Cube Organization}
\author{
{Carl Yang, Mengxiong Liu, Frank He, Jian Peng, Jiawei Han}\\
\fontsize{10}{10}\selectfont\itshape
University of Illinois, Urbana Champaign, 201 N Goodwin Ave, Urbana, Illinois 61801, USA\\
\fontsize{9}{9}\selectfont\ttfamily\upshape
\{jiyang3, mliu60, shibihe, jianpeng, hanj\}@illinois.edu
}

\begin{document}

\setlength{\floatsep}{4pt plus 4pt minus 1pt}
\setlength{\textfloatsep}{4pt plus 2pt minus 2pt}
\setlength{\intextsep}{4pt plus 2pt minus 2pt}
\setlength{\dbltextfloatsep}{3pt plus 2pt minus 1pt}
\setlength{\dblfloatsep}{3pt plus 2pt minus 1pt} 
\setlength{\abovecaptionskip}{3pt}
\setlength{\belowcaptionskip}{2pt}
\setlength{\abovedisplayskip}{2pt plus 1pt minus 1pt}
\setlength{\belowdisplayskip}{2pt plus 1pt minus 1pt}

\maketitle
\begin{abstract}
Networks are widely used to model objects with interactions and have enabled various downstream applications. However, in the real world, network mining is often done on particular query sets of objects, which does not require the construction and computation of networks including all objects in the datasets. In this work, for the first time, we propose to address the problem of \textit{query-specific network construction}, to break the efficiency bottlenecks of existing network mining algorithms and facilitate various downstream tasks. To deal with real-world massive networks with complex attributes, we propose to leverage the well-developed data cube technology to organize network objects \wrt~their essential attributes. An efficient reinforcement learning algorithm is then developed to automatically explore the data cube structures and construct the optimal query-specific networks. With extensive experiments of two classic network mining tasks on different real-world large datasets, we show that our proposed \textit{cube2net} pipeline is general, and much more effective and efficient in query-specific network construction, compared with other methods without the leverage of data cube or reinforcement learning.
\end{abstract}


\section{Introduction}
\label{sec:intro}
Networks provide a natural and generic way for modeling the interactions of objects, upon which various tasks can be performed, such as node classification \cite{yang2019neural}, community detection \cite{yang2017cone} and link prediction \cite{yang2017bi, yang2017bridging}
However, as real-world networks are becoming larger and more complex every day, various network mining algorithms need to be frequently developed or improved to scale up, but such innovations are often non-trivial, if not impossible. Moreover, the quality of networks taken by these algorithms is often questionable: Do the networks include all necessary information, and is every piece of information in the networks useful? 

While existing network mining algorithms mostly focus on more complex models for better capturing of the given network structures \cite{zhang2017weisfeiler, lyu2017enhancing, ribeiro2017struc2vec}, in this work, for the first time, we draw attention to the fact that network mining tasks are often specified on \textit{particular sets of objects of interest}, which we call \textit{queries}, and advocate for \textit{query-specific network construction}, where the goal is to construct networks that are most \textit{relevant} to the queries.

\begin{figure}[h!]
        \includegraphics[width=1.05\linewidth]{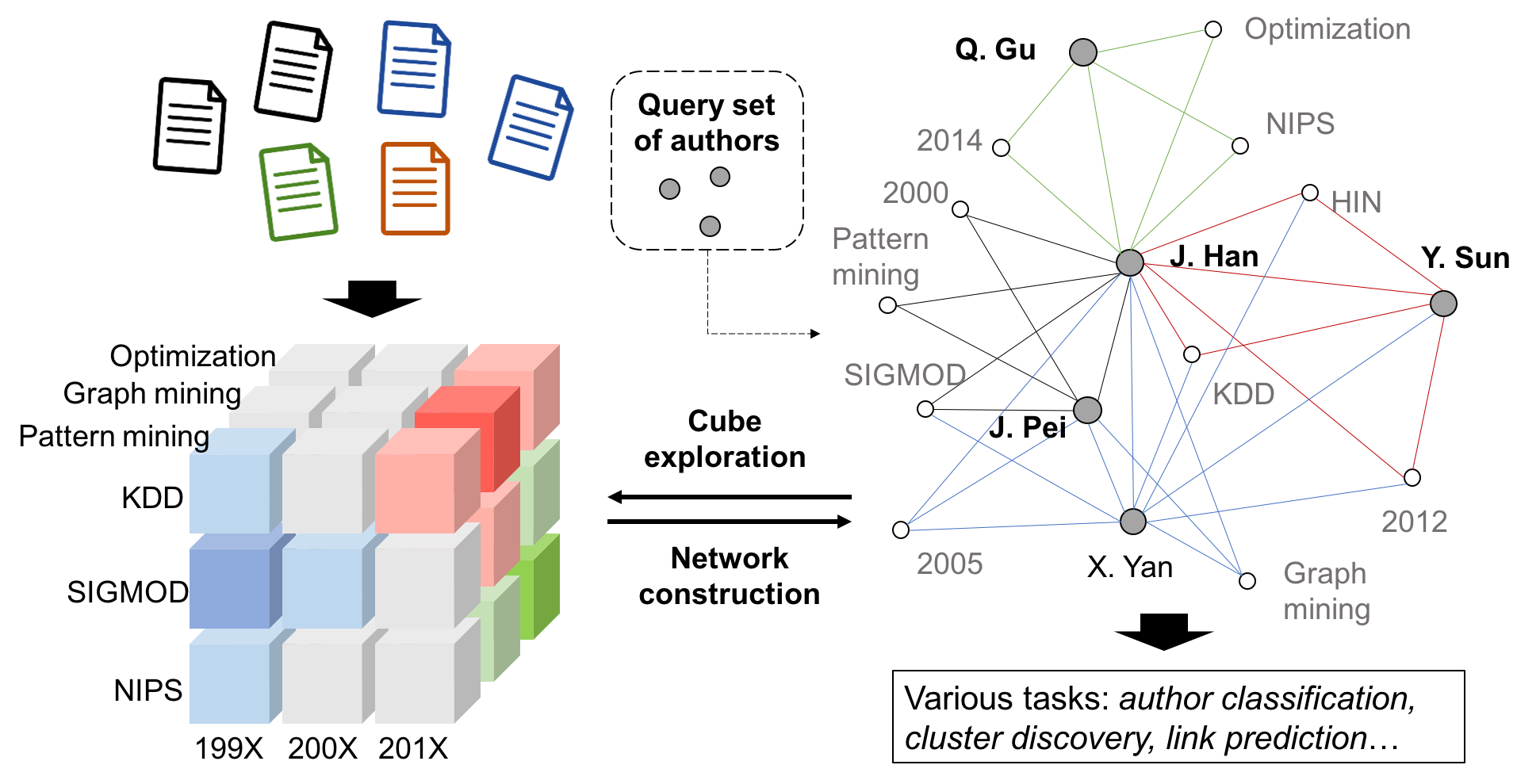}
    \caption{\textit{cube2net}: A running example on DBLP.}
    \label{fig:toy}
 \end{figure}

Under the philosophy of the well-developed technology of \textit{data cube} for large-scale data management, massive real-world networks can be partitioned into small subnetworks residing in fine-grained multi-dimensional \textit{cube cells} \wrt~their essential node properties \cite{han2011data, lin2008text}.
Assuming the proper data cubes can be efficiently constructed for particular networks automatically \cite{yang2019cubenet, zhang2018taxogen, tao2018doc2cube, tao2016multi}, we can then clearly formulate the problem of of this work as follows.

\begin{definition}{Cube-Based Query-Specific Network Construction}. 
Given a massive network with objects organized in a data cube and a query set of objects, find a set of cells, so that objects in the cells are the most relevant to the query.
\label{def:qsnc}
\end{definition}
 

Figure \ref{fig:toy} gives a toy example of query-specific network construction. 
Consider the massive author network of DBLP\footnote{http://dblp.uni-trier.de/}. The task is to find pairs of close collaborators within a particular research group. Only retaining the co-author links within the group and ignoring all outside collaborations clearly leads to significant information loss, while incorporating all co-author links in the whole network is too costly and brings in lots of useless data. 
Based on the fact that the whole network can be partitioned into fine-grained multi-dimensional cube cells like \textsf{$<$200X, KDD, Graph mining$>$}, \textsf{$<$201X, ACL, Text mining$>$}, we can look for a few subnetworks that are the most relevant to the considered group (\eg, by looking at their overlap with the group), and leverage the union of them to serve as the query-specific network.

To identify the set of cells that contain subnetworks most relevant to the query, a straightforward method is exhaustive greedy search, by expanding the network with the most relevant cell incrementally. However, this method suffers from the following two drawbacks:
\begin{enumerate}[leftmargin=15pt]
\item Since nodes in the real-networks can have various properties, there can be thousands or millions of cells, which makes exhaustive search very costly.
\item The problem of looking for a set of most relevant cells is essentially a combinatorial problem over all cells, which can not be optimally solved by the greedy algorithm.
\end{enumerate}


In this work, we propose and design \textit{cube2net}, a simple and effective reinforcement learning algorithm over the data cube structures, to efficiently find a near optimal solution for the combinatorial problem of cube-based query-specific network construction. 
In our reinforcement learning framework, the state is represented by our novelly designed continuous cell embedding vectors which capture the semantic proximities among cells in multiple dimensions, whereas the reward is designed to optimize the overall relevance between the set of selected subnetworks and the query set of objects.
In this way, \textit{cube2net} efficiently improves the utility estimation of various related cells regarding the relevance to the query by exploring each single cell, thus effectively approaching the optimal combination of relevant cells while avoiding the enumeration of all possible combinations.

The main contributions of this work are as follows:
\begin{itemize}[leftmargin=10pt]
\item For real-world large-scale network data mining, we emphasize the lack of properly constructed networks and highlight the urge of query-specific network construction.
\item We develop \textit{cube2net}, based on the well-developed data cube technology and a novel simple reinforcement learning framework, to efficiently find the most relevant set of subnetworks from the cube structure for query-specific network construction. 
\item We conduct extensive experiments using two classic data mining tasks on different real-world massive networks to demonstrate the generality, effectiveness, and efficiency of \textit{cube2net}.
\end{itemize}

\section{Problem Formulation}
\label{sec:cube}

\subsection{Preliminaries}

\subsubsection{Data Cube Basics}
Data cube is widely used to organize multi-dimensional data, such as records in relational databases and documents in text collections \cite{tao2016multi, lin2008text, zhang2018taxogen, tao2018doc2cube}. 
With well-designed cube structures, it can largely boost various downstream data analytics, mining and summarization tasks \cite{han2011data}. In the data cube, each object is assigned into a multi-dimensional \textit{cell} which characterizes its properties from multiple aspects. We call each aspect as a \textit{cube dimension} (to differentiate from vector dimension when necessary), which is formally defined as follows.

\begin{definition}{Cube Dimension}. A dimension in a data cube is defined as $\mathcal{L}^p=\{l_1^p, l_2^p, \ldots, l^p_{|\mathcal{L}^p|}\}$, where $l^p_i\in \mathcal{L}^p$ is a label in this dimension. For each particular dimension, labels are organized in either a flat or a hierarchical way.
\end{definition}

For simplicity, we focus on flat labels in this work. Based on the notion of cube dimension, we formally define a data cube in the following.

\begin{definition}{Data cube}. A data cube is defined as
$\mathcal{C}$ = $\{\mathcal{L}^1, \ldots,$ $\mathcal{L}^P, \mathcal{D}\}$, where $\mathcal{L}^p$ is the $p$-th dimension, and $\mathcal{D}$ is a collection of objects assigned to the cube cells. Each object $d\in \mathcal{D}$ can be represented as $\{l_d^1, \ldots, l_d^P\}$, where $l_d^p$ is the label of $d$ in dimension $\mathcal{L}^p$. $l_d^p$ can also be a set of labels, allowing objects to reside in multiple cells at the same time.
\end{definition}






Without loss of generality, we give a simple example of cube construction with the DBLP data, aiming to demonstrate the power of data cube for efficient data organization and exploration, which further facilitates query-specific network construction. 
In the meantime, we are aware that by no means this is the unique or best way of constructing a DBLP data cube, while we leave the exploration of more cube design choices as future work. A reasonable belief is that better cube designs can lead to more efficient and effective data organization and exploration.

Following the design of \cite{tao2018doc2cube}, we can create three cube dimensions based on paper attributes in DBLP: $\mathcal{L}^{decade}$ derived from the numerical attribute \textsf{publication year}, $\mathcal{L}^{venue}$ derived from the categorical attribute \textsf{publication venue}, and $\mathcal{L}^{topic}$ derived from textual attributes like \textsf{paper titles} and \textsf{abstracts}. 
We assign labels of the \textit{venue} dimension by \textsf{publication years}; \textit{venue} labels are assigned by conference or journal names like \textsf{KDD}, \textsf{TKDE}; \textit{topic} labels are assigned according to a latent phrase-based topic model by labels like \textsf{Neural networks}, \textsf{Feature selection}.
To compute a useful topic model, we first apply the popular AutoPhrase tool \cite{shang2017automated} to extract quality phrases from the whole corpus and represent each paper as a bag of those quality phrases. Then we apply standard LDA and assign each paper with the highest weighted topic label. 
As a consequence, each paper is arranged into multi-dimensional cells like $<$\textsf{201X, KDD, Graph mining}$>$. 

%
Each object (\eg, author) can belong to multiple cells at the same time. The co-author links may also cross different cells. To enable fast network construction, we store all links on both end objects, and add them to the network only when both end objects are selected. As a consequence, a data cube does not partition the network into isolated parts, but rather provides a fine-grained multi-dimensional index over particular subnetworks \wrt~the essential properties of objects in the dataset for efficient exploration and selection.

\subsection{Problem Definitions}
In this work, given a massive real-world network $\mathcal{N}=\{\mathcal{V}, \mathcal{E}, \mathcal{A}\}$, where $\mathcal{V}$ is the set of objects (\eg, millions of authors on DBLP), $\mathcal{E}$ is the set of links (\eg, co-author links on DBLP), and $\mathcal{A}$ is the set of object properties (\eg, favorate venues, active years, research topics for authors on DBLP), we assume a proper data cube $\mathcal{C}$ can be well constructed either automatically or by domain knowledge. Then the input and output of query-specific network construction can be defined as follows.

\textbf{Input:}
A data cube $\mathcal{C}$ that organizes and stores a massive real-world network $\mathcal{N}$ \wrt~its essential node properties $\mathcal{A}$; a query set of objects $\mathcal{Q} \in \mathcal{V}$.

\textbf{Output:}
A set of cube cells $\mathcal{S}$, which includes the most relevant subnetwork $\mathcal{M} \in \mathcal{N}$ regarding $\mathcal{Q}$. 

There are many possible ways to mathematically define the relevance between $\mathcal{M}$ and $\mathcal{Q}$. For simplicity, in this work, we intuitively require $\mathcal{M}$ to be (1) \textit{small} so that the downstream tasks can be solved efficiently, and (2) \textit{complete} so that knowledge mining over $\mathcal{Q}$ can be accurate and effective.
Based on such philosophies, we formulate the relevance function as
\begin{align}
\text{rel}(\mathcal{M}, \mathcal{Q}) = \frac{|\mathcal{M} \cap \mathcal{Q}|}{|\mathcal{M} \cup \mathcal{Q}|},
\label{eq:quality}
\end{align}
where $|\mathcal{M} \cap \mathcal{Q}|$ is large when $\mathcal{M}$ is \textit{complete} by covering more objects in $\mathcal{Q}$, and $|\mathcal{M} \cup \mathcal{Q}|$ is small when $\mathcal{M}$ is \textit{small} and also overlaps much with $\mathcal{Q}$. $r(\mathcal{M}, \mathcal{Q}))$ reaches the optimal value $1$ when the $\mathcal{M}$ exactly covers all objects in $\mathcal{Q}$, which only happens when the union of objects in a certain set of cells are exactly the same as $\mathcal{Q}$. 



%
%

\section{cube2net}
\label{sec:rl}
\subsection{Motivations}
The combinatorial property of query-specific network construction lies in the process of selecting a particular set of objects from the very large set of all objects to construct the query-relevant subnetwork. 
Without the guidance of node properties, such a process can only be done by link-based heuristic search algorithms \cite{gionis2017bump, sozio2010community}. 

In this work, however, by leveraging the powerful technique of data cube, we can organize the massive network \wrt~its essential node properties and search over the cube cells instead of individual objects to construct the query-specific network. However, even with such efficient data organization, we are still facing the challenge of selecting over a large number of cube cells. For example, on DBLP, after removing the empty ones, we still have about 80K cells, which leads to a total number of $2^{80K}$ possible combinations. 




Reinforcement learning has been intensively studied for solving complex planning problems with consecutive decision makings, such as robot control and human-computer games \cite{silver2017mastering, sutton2000policy}. 
Recently, there are several approaches to tackle the combinatorial optimization problems over network data by reinforcement learning \cite{bello2016neural, khalil2017learning}, which are shown even more efficient than other advanced neural network models \cite{vinyals2015order}.
Motivated by their success, we aim to leverage reinforcement learning to efficiently approximate the optimal combinatorial solution for query-specific network construction. 

\subsection{Framework Formulation}
In this subsection, we demonstrate our reinforcement learning framework. 
Given a data cube $\mathcal{C}$ with $n$ cells, a partial solution is represented as $\mathcal{S}=(c_1, \ldots, c_{\mid \mathcal{S} \mid}), c_i \in \mathcal{C}$; $\overline{\mathcal{S}}=\mathcal{C}\setminus \mathcal{S} $ is the set of candidates to be selected conditional on $\mathcal{S}$. 

In \cite{khalil2017learning}, an approximate solution is proposed by using a popular heuristic, namely a greedy algorithm. Their algorithm constructs a solution by sequentially and greedily adding nodes that maximize the evaluation function of a partial solution. The evaluation function is trained on small instances (\eg, 50 nodes) with deep Q-learning and the solution is tested on larger instances (\eg, 1000 nodes) for the performance of generalization. Despite their success on several graph combinatorial optimization tasks, this heuristic does not scale to our problem. 
In our scenario, the action space is of the size $|\mathcal{C}|$, which means in every step of the deep Q-learning, thousands of selections over the cells need to be explored and exploited, making the learning process computationally intractable.

Motivated by \cite{bello2016neural}, our neural network architecture models a stochastic policy $\pi(a \mid S, \mathcal{C})$, where $a$ is the action of selecting the next cell to be added from the candidate set $\overline{\mathcal{S}}=\mathcal{C}\setminus \mathcal{S}$ and $S$ is the state of the current partial solution $\mathcal{S}$. Our target is to learn the parameters of the policy network such that $p(S \mid \mathcal{C})$  assigns a high probability to $S$ that has high quality $q(S)$ given all cells $\mathcal{C}$. We use the chain rule to factorize the probability of a solution as follows:

\begin{equation}
    p(S\mid \mathcal{C})=\prod_{i=1}^{m} \pi(S(i) \mid S(<i), \mathcal{C}),
\end{equation}
where $m$ is the size of the solution. Note that, we refer $\pi(a \mid \mathcal{S})$ to $\pi_\Theta(a \mid S, \mathcal{C})$, where $\Theta$ is the set of parameters inside the policy $\pi$.

\subsubsection{Representation}
As we just discussed, the number of possible states is exponential to the number of cells (\ie, $2^n$). It is impossible to model the states discretely with a look-up table. 
Therefore, we propose the novel technique of cell embedding, which captures the semantic proximities among cells. 
Particularly, for each cell $c\in \mathcal{C}$, a $\kappa$-dimensional embedding vector $\mathbf{u}_c$ is computed to capture $c$'s multi-dimensional semantics regarding the essential properties of the objects it contains (more details are in \textit{Appendix B}). 

Now we study how to naturally leverage such cell embedding for the design of a reinforcement learning algorithm that efficiently explores the cube structures. Specifically, we need to design an appropriate reinforcement learning state representation to encode the currently selected cells in the subset $\mathcal{S}$. 
For model simplicity and easy optimization, we directly compute the state representation as a summation of the embeddings of selected cells in $\mathcal{S}$, \ie, $\mathbf{S}=\sum_{c\in S}\mathbf{u}_c$. 
With this simple form, we theoretically show that we allow the reinforcement learning agent to efficiently explore the cube structures, by simultaneously estimating the utility of semantically close cells at each state. Particularly, we prove the following theorem. 
\begin{theorem}
Given our particular way of state representation, at each state $S$, the effect of the reinforcement learning agent in exploring a particular cell $c\in \mathcal{C}$ is similar to that of exploring any cell $c'\in\mathcal{N}_\xi (c)$, when $\xi$ is sufficiently small, in the sense that the output of the actor (critic) network is similar, \ie, $|\mathbf{f}(S, c)-\mathbf{f}(S,c')|\leq \eta\xi$, and the gradients are similar too, \ie, $|| \nabla \mathbf{f}(S,c) - \nabla \mathbf{f}(S,c') ||^2_2 \leq  \zeta\xi$, where $\mathbf{f}=\mu$ or $\nu$.
\label{th:rl}  
\end{theorem}
Moreover, according to Theorem \ref{th:rl}, it is easy to arrive at the following Corollary.
\begin{corollary}
Given our particular way of state representation, the effect of the reinforcement agent of exploring a particular cell $c\in \mathcal{C}$ in a particular state $S$ is similar to that of exploring any cell $c'\in\mathcal{N}_{\xi_1}(c)$ at any state $S'\in\mathcal{N}_{\xi_2}(S)$, when both $\xi_1$ and $\xi_2$ are sufficiently small.  
\label{co:rl}
\end{corollary}
For detailed proofs, please refer to \textit{Appendix A}.



Continue with our toy example on DBLP, where the deeply colored cells are directly sampled and estimated by the reinforcement learning agent during the exploration. 
According to Theorem \ref{th:rl} and Corollary \ref{co:rl}, assuming the lightly colored cells are those semantically close to the deeply colored ones, their utilities regarding the quality function in similar states are also implicitly explored and evaluated.
In this way, the reinforcement learning algorithm efficiently avoids the exhaustive search over all cube cells and their possible combinations.

In practice, we find that long trajectories leading to large constructed networks usually do not result in efficient model inference and satisfactory task performance. 
Therefore, we set $\delta(S)=\{|\mathcal{S}|<=m\}$, which terminates the selection of more cells if the number of selected cells reaches $m$. 
In Section \ref{sec:exp}, we will study the impact of $m$ on the performance of our algorithm.

\subsubsection{Reinforcement Learning Formulation} 
The components in our reinforcement learning framework are defined as follows.

\begin{enumerate}[leftmargin=15pt]
\item \textbf{State:} A state $S$ corresponds to a set of cells $\mathcal{S}$ we have selected. Based on previous discussion, a state is represented by a $\kappa$-dimensional vector $\mathbf{S}=\sum_{c \in \mathcal{S}} \mathbf{u}_c$.

\item \textbf{Action:} An action $a$ is a cell $c$ in $\overline{\mathcal{S}}=\mathcal{C}\setminus \mathcal{S}$. We will cast the details of the actions later.

\item \textbf{Reward:} The reward $r$ of taking action $a$ at state $S$ is 
\begin{equation}
    r(S, a)=q(S')-q(S),
    \label{eq:reward}
\end{equation}
where $S':=(S, c)$. Given Eq.~\ref{eq:reward}, maximizing the cumulative reward $\sum_{i=1}^{n-1} r(S_i, c_i)$ is the same as maximizing the quality function $q(S)$.

\item \textbf{Transition:} The transition is deterministic by simply adding the cell $c$ we have selected according to action $a$ to the current state $S$. Thus, the next state $S':=(S,c)$.
\end{enumerate}

According to our discussion in Section \ref{sec:cube}.2, the quality function $q$ is set to $q(S)=\text{rel}(\mathcal{M}, \mathcal{Q})$, where $\mathcal{M}$ is the set of objects in the selected cells $\mathcal{S}$ at state $S$.

Continue with our toy example in Figure \ref{fig:toy}, where $\mathcal{Q}$ is the query set of authors including two data mining researchers,\eg, \textsf{Y. Sun} and \textsf{J. Pei}. Assume the two authors are in the same cube cell \textsf{$<$201X, KDD, Data mining$>$}. Consider other authors in DBLP like \textsf{X. Yan}, who is also in the cell \textsf{$<$201X, KDD, Data mining$>$}, and \textsf{Q. Gu}, who is in another cell \textsf{$<$201X, NIPS, Optimization$>$}. Our quality function $q(S)$ will prefer to include \textsf{X. Yan} into the constructed network, who is indeed more relevant to both \textsf{Y. Sun} and \textsf{J. Pei}, and such inclusion will likely benefit the proximity modeling among the particular authors in $\mathcal{Q}$.

Note that, although the toy example looks simple, when we consider $\mathcal{Q}$ including thousands of objects distributed in multiple cells, as well as their relevant objects distributed in thousands of cells, the selection of the most relevant cells becomes a complex combinatorial problem.
We stress that the main contribution of this work is to provide a general framework for efficient network construction, and it is trivial to plug in different quality functions when certain properties of the constructed networks are required or preferred. 
In Section \ref{sec:exp}, we show the power of our framework for multiple network mining tasks with the simple quality function $q$, and we leave the exploration of more theoretically sound or task-specific ones to future work.

\subsection{Learning Algorithm}
We illustrate how reinforcement learning is used to learn the policy parameters $\Theta$. 
It is impossible to directly apply the commonly used Q-learning to our task because finding the greedy policy requires optimization over $a_t$ at every timestep $t$, which is too slow to be practical with large action spaces \cite{lillicrap2015continuous}. To this end, we adopt continuous policy gradient as our learning algorithm. We  design our actions to be in the continuous cell embedding space. At each time, the policy outputs a continuous action vector and then we retrieve the closest cell by comparing the action vector with the cell embeddings.

\begin{figure}[h!]
    \includegraphics[width=0.98\linewidth]{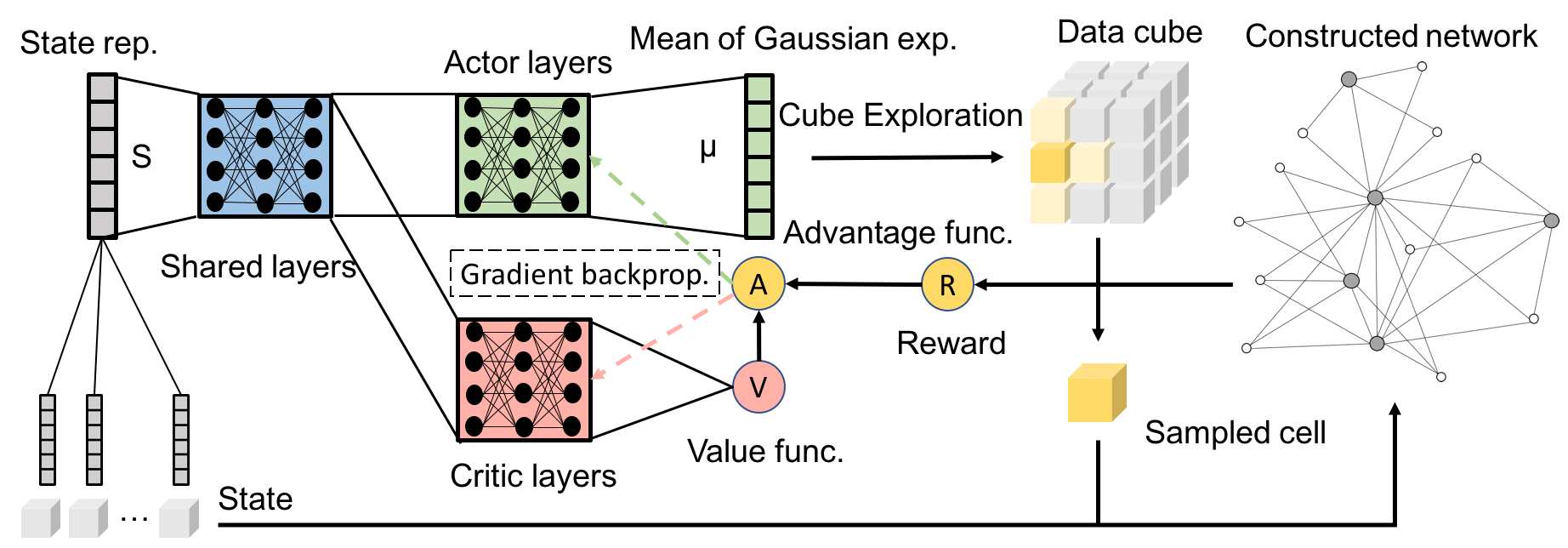}
    \caption{The overall learning paradigm of cube2net.}
    \label{fig:learning}
\end{figure}

Figure \ref{fig:learning} illustrates our learning algorithm. More formally, our policy $\pi(a \mid S)$ takes the state representation $\mathbf{S}$ as the input and produces an action as the output. Instead of using Boltzmann exploration (\eg, softmax) to choose a discrete action $a=c$ from $\overline{\mathcal{S}}$, we project the actions to the $\kappa$-dimensional cell embedding space. Based on the property that semantically similar cells are close in the embedding space, a Gaussian exploration with mean $\mu \in \mathbb{R}^\kappa$ and covariance matrix $\Sigma \in \mathbb{R}^{\kappa \times \kappa}$ is used in our policy \cite{lillicrap2015continuous} as follows.

\begin{align}
    \pi(a \mid S, \{\mu, \Sigma \})=\frac{1}{\sqrt{2\pi \left| \Sigma \right|}} \exp \left(-\frac{1}{2}(a-\mu)^T \Sigma^{-1} (a-\mu)\right),
\end{align}
\begin{align}
    \nabla_\mu \log \pi (a \mid S, \{\mu, \Sigma \})&=(a-\mu)^T\Sigma^{-1},\\
    \nabla_\Sigma \log \pi (a \mid S, \{\mu, \Sigma \})&=\frac{1}{2}\left(\Sigma^{-1} (a-\mu) (a-\mu)^T\Sigma^{-1} - \Sigma^{-1} \right).
\end{align}
A neural network $\mu_\Theta(S)$ is used as our actor network to produce a Gaussian mean vector $\mu$ as part of the action. A learnable parameter $\sigma$ is used to model the Gaussian variance vector as $\Sigma=\sigma^2 I$.

\subsubsection{Policy Gradient}
Since our goal is to find a high-quality solution $S$, our training objective is 
\begin{equation}
U(\Theta \mid C) = \mathbb{E}_{\tau \sim \pi_\Theta(S\mid C)} R(\tau),
\label{obj_pg}
\end{equation}
where $\tau$ denotes a state-action trajectory, \ie, $S_0, a_0, \cdots, S_T, a_T$. Each $\tau$ corresponds to a selected set of cube cells. $T(=m)$ is the length of the trajectory, and $R(\tau)=\sum_{t=0}^T r(S_t, a_t)$ is the reward.
Then we derive the gradient estimator based on policy gradient as
\begin{align}
\nabla_\Theta U=&\frac{1}{\alpha} \sum_{i=1}^{\alpha} \sum_{t=0}^{T-1} \nonumber\\
&\nabla \log \pi_\Theta (a_t^{(i)}\mid S_t^{(i)}) \left(\sum_{k=t}^{T-1}r(S_k^{(i)}, a_{k}^{(i)}) -\nu_\Theta(S_k^{(i)})\right),
\end{align}
where $\alpha$ is the number of trajectories, $\nu_\Theta(S)$ is the value function. 

In a general form, we have the gradient estimator as
\begin{equation}
\hat{g}=\mathbb{\hat{E}}_t[\nabla_\Theta \log \pi_\Theta (a_t \mid S_t) \hat{A}_t],
\end{equation}
where $\hat{A}_t$ is the advantage function at timestep $t$. $\mathbb{\hat{E}}_t$ denotes the empirical average over a mini-batch of samples in the algorithm that alternates between sampling and optimization using the proximal policy gradient algorithm \cite{schulman2017proximal}.


To compute the advantage function $\hat A_t$, we use a T-step advantage function estimator as 
\begin{align*}
\hat{A}_t=-\nu_\Theta(S_t)+r(S_t, a_t)+\gamma r(S_{t+1}, a_{t+1}) + \cdots \\+ \gamma^{T-t+1}r(S_{T-1}, a_{T-1})+\gamma^{T-t}\nu_\Theta(S_T),
\end{align*}
where $\gamma$ is a clipping parameter. $\nu_\Theta(S)$ is learned by minimizing the loss $L^{\text{VF}}(\Theta)=\|\nu_\Theta(S_t) - \nu^{\text{target}}(S_t)\|_2^2$, where $V^{\text{target}}$ is the advantage function minus the value function.

\subsubsection{Neural Architecture}
Our actor network (policy network) $\mu_\Theta(S)$ and critic network (value function) $\nu_\Theta(S)$ shown in Figure \ref{fig:learning} are both fully connected feedforward neural networks, each containing four layers including two hidden layers of size $H$, as well as the input and output layers. 
\textsf{Sigmoid} is used as the activation function for each layer, and the first hidden layer is shared between two networks. Both networks' inputs are $\kappa$-dimensional cell embeddings. The output of the actor network $\mu_\Theta(S)$ is a $\kappa$-dimensional vector $ \mathbb{R}^\kappa$ and the output of critic network $\nu_\Theta(S)$ is a real number.

As for model learning, in each iteration, our algorithm samples a set $\Phi$ of $\alpha$ trajectories (series of cells) of length $T$ using the current policy $\pi_\Theta(a \mid S)$, and constructs the surrogate loss and value function loss based on $\Phi$. Mini-batch SGD is then used to optimize the objectives for $\beta$ epochs, where the model parameters in $\Theta$ are updated by Adam \cite{kingma2014adam}.

\subsection{Computational Complexity}
We theoretically analyze the complexity of our algorithm. Consider the problem of selecting $m$ cells from a data cube $\mathcal{C}$ with $n$ cells ($m \ll n$).
During each step of model training, \textit{cube2net} generates a target mean $\mu_\Theta$ in constant time and then selects a cell from $C$ that is the closest to $\mu_\Theta$. 
Since computing the quality function and updating the neural network model based on particular trajectories take constant time and \textsf{argmax} takes logarithm time when parallelized on GPU, the overall complexity of training and planning with \textit{cube2net} is $O(\alpha \beta cm \log n)$. $\alpha$, $\beta$ are number of trajectories and epochs, which are both set to 40 in our experiments; $c$ is the constant time in computing the quality function. The time for model inference is ignorable compared with model training.

As a comparison, a random algorithm takes $O(1)$ time to select each cell and has an overall complexity of $O(m)$. 
A greedy algorithm at each of the $m$ steps needs to consider all $n$ actions and run the quality function on each of them to get the maximum, which has an overall complexity of $O(cmn)$.
In practice, computing the quality function cannot be parallelized on GPU and always takes the most significant time. Also, the time is not truly constant-- it largely favors \textit{cube2net} as it constructs small networks by looking for the globally optimal combinations of relevant cells.

Note that, the novelty of our reinforcement learning algorithm lies in its unique leverage of cube structures and cell embedding, as well as the efficient Gaussian exploration policy. However, the training of it is rather routine, which as we will also show in Section \ref{sec:exp}, is very efficient and does not require heavy parameter tuning.

\section {Experimental Evaluation} 
\label{sec:exp}
In this section, we evaluate the effectiveness and efficiency of our proposed \textit{cube2net} pipeline on two large real-world datasets from very different domains, \ie, DBLP, an academic publication collection and Yelp\footnote{https://www.yelp.com/}, a business review platform. We compare different network construction methods for two well-studied tasks on each of them-- author clustering on DBLP and business-user link prediction on Yelp, to show the power and generality of \textit{cube2net}. 

\subsection{Author Clustering on DBLP}
We first study the author networks of DBLP, which we have been using as the running example throughout this paper. Particularly, we focus on the problem of author clustering, since it is a well-studied problem with two sets of public standard evaluation labels \cite{sun2011pathsim}, which are not directly captured by the network structures or attributes. Since the essential challenge of network mining is to capture object proximities (\eg, using embeddings), we take an \textit{embedding-clustering} approach and provide a comprehensive evaluation of the embedding quality regarding clustering. 
As shown in the experimental results of various embedding techniques \cite{perozzi2014deepwalk, tang2015line, grover2016node2vec}, the quality of network embeddings are often consistent across different network mining tasks. Therefore, superior performances on author clustering with embedding can be a good indicator to the general high quality of constructed networks towards various network mining tasks.

\subsubsection{Experimental Settings}
{\flushleft \bf Dataset.}
We use the public DBLP dataset V10\footnote{https://aminer.org/citation} collected by \cite{tang2008arnetminer}.
The basic statistics of the dataset are shown in Table \ref{tab:stat}.

\begin{table}[h!]
\centering
 \begin{tabular}{|c|ccc|}
   \hline
{\bf Dataset}&{\bf \#papers}&{\bf \#authors}&{\bf \#links}\\
  \hline
{\bf DBLP} & 3,079,007 & 534,407 & 28,347,138 \\
\hline
\hline
{\bf Dataset}&{\bf \#businesses}&{\bf \#customers}&{\bf \#links}\\
\hline
{\bf Yelp}& 174,567 & 1,326,101 & 5,261,669\\
\hline
 \end{tabular}
 \caption{ \label{tab:stat}\textbf{Statistics of the two public datasets we use.}}
\end{table}

The DBLP dataset contains semi-structured scientific publications, with their corresponding authors, years, venues and textual contents. As discussed in Section \ref{sec:cube}, a simple data cube with dimensions \textit{year}, \textit{venue} and \textit{topic} is constructed, with each cell holding the corresponding \textit{authors}. 

For evaluating the task of author clustering, we use two sets of authors with ground-truth labels regarding research groups and research areas published by \cite{sun2011pathsim}. They are commonly used as ground-truth author labels for evaluating network classification and clustering algorithms on DBLP. The smaller set includes 4 labels of research groups on 116 authors, and the larger set includes 4 labels of research areas on 4,236 authors. We separately use them as the query sets of objects $\mathcal{Q}$ in two sets of experiments.


{\flushleft \bf Baselines.}
We aim to improve network mining by constructing the proper network that contains only the query set of objects $\mathcal{Q}$ and the most relevant additional objects $\mathcal{Q}^+$. Since we are the first to study the problem of query-specific network construction, we design a comprehensive group of baselines that can potentially choose a relevant set of objects $\mathcal{Q}^+$ to add to $\mathcal{Q}$ in the following.
\begin{itemize}[leftmargin=15pt]
\item \textbf{NoCube}: Without a data cube, this method adds no additional object to the smallest network with only $\mathcal{Q}$ as nodes.
\item \textbf{NoCube+}: Without a data cube, this method adds all objects that are directly linked with $\mathcal{Q}$ to the network.
\item \textbf{NoCube++}: Without a data cube, this method adds all objects that are within two steps away from $\mathcal{Q}$ to the network. 
\item \textbf{MaxDisc} \cite{gionis2017bump}: Without a data cube, this method leverages a heuristic search algorithm to find the largest network connected component that contains $\mathcal{Q}$.
\item \textbf{CubeRandom}: With a proper data cube, this method adds objects in $m$ random cells to the network.
\item \textbf{CubeGreedy}: With the same data cube, this method searches through all cells for $m$ times, and greedily add the objects in one cell at each time to optimize the quality function of Eq.~\ref{eq:quality}.
\end{itemize}

{\flushleft \bf Parameter settings.}
When comparing \textit{cube2net} with the baseline methods, we set its parameters to the following default values: 
For cube construction, we set the number of topics in LDA to 100 and we filter out cells with less than 10 objects;
for reinforcement learning, we empirically set the size of hidden layers $H$ to 256, the length of trajectories $T$ (\ie, the number of added cells $m$) to $20$, the sample size $\alpha$ and the number of epochs $\beta$ both to 40, and the clipping parameter of value function $\gamma$ to 0.99.
In addition to these default values, we also evaluate the effects of different parameters such as the length of trajectories $T$ to study their impact on the performance of the three cube-based algorithms. 

{\flushleft \bf Evaluation protocols.}
We aim to show that \textit{cube2net} is a general pipeline that breaks the bottleneck of network mining by constructing query-specific networks. 
Therefore, we evaluate the compared algorithms for network construction through two perspectives: effectiveness and efficiency. 

We study the effectiveness by measuring the performance of downstream applications on the constructed network $G$. 
For author clustering, we firstly compute a network embedding $E$ on $G$, which converts the network structure into distributed node representations.
To provide a comprehensive evaluation, we use three popular algorithms, \ie, \textit{DeepWalk} \cite{perozzi2014deepwalk}, \textit{LINE} \cite{tang2015line} and \textit{node2vec} \cite{grover2016node2vec} for this step. 
We then feed $E$ into the standard $K$-means.

To evaluate the clustering results, we compute the \textit{F1 similarity} (F1), \textit{Jaccard similarity} (JC) and \textit{Normalized Mutual Information} (NMI) \wrt ground-truth labels. Detailed definitions of the three metrics can be found in \cite{liu2015community}.

The efficiency of network construction is reflected by the \textit{time} spent on constructing the network, as well as the \textit{size} of the constructed network (in terms of \#nodes and \#edges), which can further influence the runtime of the network mining algorithms.
All runtimes are measured on a server with one GeForce GTX TITAN X GPU and two Intel Xeon E5-2650V3 10-core 2.3GHz CPUs.

\subsubsection{Performance Comparison with Baselines}
\begin{table*}[t!]
\vspace{-10pt}
\centering
\scriptsize
 \begin{tabular}{|c||ccc|ccc|ccc|}
 \hline
\multirow{3}{*}{\begin{tabular}{@{}c@{}} {\bf Net.~Con.} \\ {\bf Algorithm}\end{tabular}}&\multicolumn{9}{c|}{\bf Effectiveness on the smaller set of authors}\\
   \cline{2-10}
&\multicolumn{3}{c|}{F1}&\multicolumn{3}{c|}{JC} &\multicolumn{3}{c|}{NMI} \\
 \cline{2-10}
  & {\it DeepWalk}& {\it LINE}& {\it node2vec} & {\it DeepWalk}& {\it LINE}& {\it node2vec}& {\it DeepWalk}& {\it LINE}& {\it node2vec} \\
\hline
NoCube & 0.7034 & 0.5620 & 0.6195 & 0.4828 & 0.4621 & 0.5179  & 0.4013 & 0.3642 & 0.3976\\
\hline
NoCube+ & 0.6559 & 0.5263 & 0.5812 & 0.4359 & 0.4178 & 0.4763 & 0.2928 & 0.2814 & 0.3058\\
\hline
NoCube++ & 0.6794 & 0.5247 & 0.5874 & 0.4678 & 0.4155 & 0.4616 & 0.3799 & 0.3152 & 0.3654\\
\hline
MaxDisc & 0.7162 & 0.5423 & 0.6299 & 0.4983 & 0.4437 & 0.5305 & 0.4052 & 0.3246 & 0.4122\\
\hline
CubeRandom & 0.5839 & 0.5087 & 0.5748 & 0.4681 & 0.4194 & 0.5069 & 0.3693 & 0.3105 & 0.3930\\
\hline
CubeGreedy & 0.7445 & 0.5988 &  0.6432 & 0.5492 & 0.4850 & 0.5343 & 0.3981 & 0.3369 & 0.4086\\
\hline
cube2net & \textbf{0.7628} & \textbf{0.6295} & \textbf{0.6913} & \textbf{0.5720} & \textbf{0.5312} & \textbf{0.5834} & \textbf{0.4196} & \textbf{0.3784} & \textbf{0.4517}\\
\hline
\multirow{3}{*}{\begin{tabular}{@{}c@{}} {\bf Net.~Con.} \\ {\bf Algorithm}\end{tabular}}&\multicolumn{9}{c|}{\bf Effectiveness on the larger set of authors}\\
   \cline{2-10}
&\multicolumn{3}{c|}{F1}&\multicolumn{3}{c|}{JC} &\multicolumn{3}{c|}{NMI} \\
 \cline{2-10}
  & {\it DeepWalk}& {\it LINE}& {\it node2vec} & {\it DeepWalk}& {\it LINE}& {\it node2vec}& {\it DeepWalk}& {\it LINE}& {\it node2vec} \\
\hline
NoCube & 0.4336 & 0.3044 & 0.3708 & 0.2541 & 0.1759 & 0.2195  & 0.0809 & 0.0426 & 0.0439\\
\hline
NoCube+ & 0.5207 & 0.3212 & 0.5113 & 0.3022 & 0.1773 & 0.3362 & 0.1105 & 0.0454 & 0.0996\\
\hline
NoCube++ & 0.5515 & 0.3261 & 0.4984 & 0.3989 & 0.1725 & 0.3465 & 0.2174 & 0.0436 & 0.1086\\
\hline
MaxDisc & 0.5859 & 0.3282 & 0.5619 & 0.4249 & 0.1950 & 0.4071 & 0.2120 & 0.0484 & 0.2278\\
\hline
CubeRandom & 0.4018 & 0.2972 & 0.3416 & 0.2158 & 0.1543 & 0.1766 & 0.0620 & 0.4021 & 0.0376\\
\hline
CubeGreedy & 0.6125 & 0.3509 &  0.5768 & 0.4526 & 0.1954 & 0.4186 & 0.2513 & 0.0595 & 0.2224\\
\hline
cube2net & \textbf{0.6447} & \textbf{0.3718} & \textbf{0.6214} & \textbf{0.4865} & \textbf{0.2288} & \textbf{0.4623} & \textbf{0.2597} & \textbf{0.0646} & \textbf{0.2418}\\
\hline
 \end{tabular}
 \caption{ \label{tab:dblp-s-perf}\textbf{Effectiveness of query-specific network construction for clustering query set of authors.}}
\end{table*}

\begin{table*}[t!]
\centering
\scriptsize
 \begin{tabular}{|c||cccc|cc|}
   \hline
\multirow{3}{*}{\begin{tabular}{@{}c@{}} {\bf Net.~Con.} \\ {\bf Algorithm}\end{tabular}}&\multicolumn{6}{c|}{\bf Efficiency on the smaller set of authors} \\
   \cline{2-7}
&\multicolumn{4}{c|}{Computation Time} & \multicolumn{2}{c|}{Network Size} \\
 \cline{2-7}
 &  {\it DeepWalk} & {\it LINE} & {\it node2vec} & Net. Con. & \#nodes & \#edges \\
\hline
NoCube &1s & 1s & 1s & 2s &116 & 318 \\
\hline
NoCube+ &25s & 58s & 65s & 5s & 3,853 & 32,882 \\
\hline
NoCube++ &732s & 652s & 904s& 62s  & 55,724 & 540,780 \\
\hline
MaxDisc & 1s & 2s & 3s & 794s & 140 & 440 \\
\hline
CubeRandom &56s & 59s & 62s & 4s & 2,512 & 24,578 \\
\hline
CubeGreedy &28s & 43s & 32s & 3,082s & 1,464 & 14,892 \\
\hline
cube2net & 7s & 11s & 10s & 296s & 526 & 2,953 \\
\hline
\multirow{3}{*}{\begin{tabular}{@{}c@{}} {\bf Net.~Con.} \\ {\bf Algorithm}\end{tabular}}&\multicolumn{6}{c|}{\bf Efficiency on the larger set of authors} \\
   \cline{2-7}
&\multicolumn{4}{c|}{Computation Time} & \multicolumn{2}{c|}{Network Size} \\
 \cline{2-7}
 &  {\it DeepWalk} & {\it LINE} & {\it node2vec} & Net. Con. & \#nodes & \#edges \\
\hline
NoCube &11s & 74s & 23s & 3s & 4,236 & 4,678 \\
\hline
NoCube+ &270s & 127s & 2,147s & 84s & 74,459 & 120,803 \\
\hline
NoCube++ &2,450s & 1,514s & 8,485s & 1,128s  & 434,941 & 1,372,892 \\
\hline
MaxDisc & 18s & 78s & 113s & 3,390s & 4,718 & 7,004 \\
\hline
CubeRandom &104s & 108s & 430s & 4s & 21,126 & 31,481 \\
\hline
CubeGreedy &62s & 96s & 208s & 4,194s & 10,046 & 16,259 \\
\hline
cube2net & 23s & 83s & 92s & 314s & 6,842 & 12,497 \\
\hline
 \end{tabular}
 \caption{ \label{tab:dblp-s-time}\textbf{Efficiency of query-specific network construction for clustering query set of authors.}}
\end{table*}

We quantitatively evaluate \textit{cube2net} against all baselines on clustering the query set of authors. 
Tables \ref{tab:dblp-s-perf} and \ref{tab:dblp-s-time} comprehensively show the effectiveness and efficiency of compared algorithms.

The scores are deterministic for most algorithms and have very small standard deviations on the others except for \textit{CubeRandom}. We run \textit{CubeRandom} and \textit{cube2net} for 10 times to take the average. The improvements of \textit{cube2net} over other algorithms all passed the paired t-tests with significance value $p<0.01$.

As for effectiveness, the network constructed by \textit{cube2net} is able to best facilitate network mining around the query, particularly, author network clustering in this experiment, under the three metrics computed on standard $K$-means clustering results after network embedding. 
As can be observed, (1) blindly adding neighbors into the network without a cube organization or randomly adding cells can hurt the task performance; (2) with a proper cube structure, greedily adding cells \wrt our quality function can significantly boost the task performance; however, (3) the performance of \textit{CubeGreedy} is still inferior to \textit{cube2net}, which confirms our arguments that the task of network construction is essentially a combinatorial problem, which requires a globally optimal solution that can be efficiently achieved only by reinforcement learning.

Note that, we observe that blindly enriching the network with neighbor nodes does not always hurt the performance. 
To better understand it, we look into the data and find that links among the 116 labeled authors are much denser than the 4,236 authors (\ie, $4.77\%$ \textit{v.s.~} $0.05\%$). It indicates adding neighbors are more helpful when the existing interactions among the query set itself are more sparse. Particularly, we find that among the 4,236 authors, 930 were dangling in the original network, but can be bridged into the constructed network with other queried authors by adding their direct neighbors. 

As for efficiency, (1) without cube organization, the network can easily get too large, which requires significant network construction time, and leads to long runtimes of network mining algorithms; 
(2) the sizes of constructed networks are much more controllable with a proper cube organization, because we can easily set the number of cells to add; (3) even with a proper cube, greedily searching the cube at each step to select the proper cells is extremely time-consuming-- on the contrary, \textit{cube2net} efficiently explores the cube structures with reinforcement learning and is able to find the particular set of cells to construct the most relevant subnetwork, which also makes the downstream network mining more efficient.

Comparing the results on the two sets of query objects of different sizes, we further find that,
(1) as the query set of authors becomes larger, blindly bringing in neighbors leads to much larger networks, and subsequently much longer runtimes of network mining algorithms. Such low efficiency is exactly what we want to avoid by aiming at query-specific network construction in this work;
(2) when the query set becomes much larger, the runtimes of the cube-based algorithms only increase a little, since they still work on the same well organized cube structure by evaluating the utility of cells rather than individual nodes, indicating the power of the data cube organization.

Note that, since \textit{cube2net} needs to learn the policy every time given a new set of objects, the network construction time we report for \textit{cube2net} is a sum of three runtimes: model training, model inference and network link retrieval. In this case, the comparison to other algorithms is fair, whose network construction time is a sum of node selection and link retrieval. 

\begin{figure}[h!]
\centering
\subfigure[DeepWalk - $F1$]{
\includegraphics[width=0.23\textwidth]{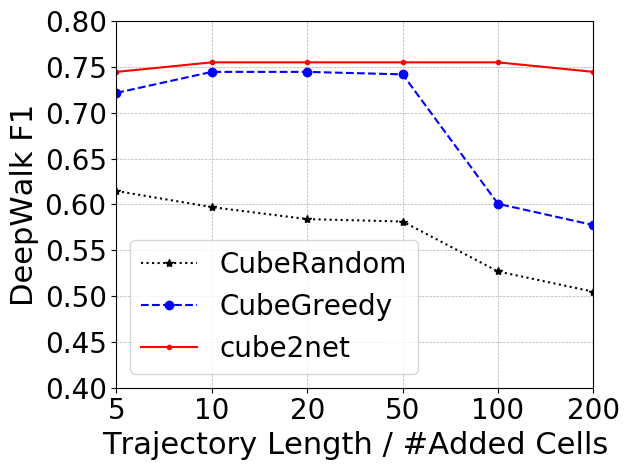}}
\hspace{-8pt}
\subfigure[node2vec - $F1$]{
\includegraphics[width=0.23\textwidth]{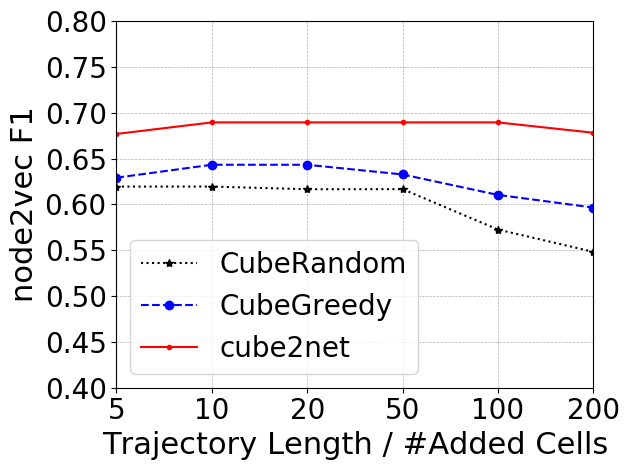}}
\caption{\textbf{Comparing different trajectory lengths.}}
\label{fig:ncell}
\end{figure}

In Figure \ref{fig:ncell}, we also show the performances of the three cube-based algorithms as we consider different numbers of cells to add into the network with the small set of 116 labeled authors. Although particular network mining algorithms might prefer different network structures, the performances of three compared network construction algorithms vary in very similar trends: (1) \textit{cube2net} maintains the best and most robust performance in all cases; (2) \textit{CubeGreedy} performs closely with \textit{cube2net} when the number of added cells is small, but gets worse as more cells are considered and the combinatorial property of the problem emerges; (3) \textit{CubeRandom} keeps getting worse as more random cells are added into the network.

\subsubsection{Case Study}
\begin{table}[h!]
\centering
\scriptsize
 \begin{tabular}{|c|cccc|}
   \hline
{\bf Authors} & \multicolumn{4}{c|}{\bf Cells selected for the given set of authors}\\
  \hline
\multirow{6}{*}{\begin{tabular}{@{}c@{}} {\sf J. Han } \\ {\sf D. Blei } \\ {\sf D. Roth } \\ {\sf J. Leskovec } \\ {\sf ... } \end{tabular}}& {\bf Rank} & {\bf Decade} & {\bf Venue} & {\bf Topic}\\
   \cline{2-5}
  & 1 & {\sf 200X} & {\sf KDD} & {\sf Recommender systems}\\
   \cline{2-5}
   & 2 & {\sf 200X} & {\sf AAAI} & {\sf Knowledge discovery}\\
   \cline{2-5}
   & 3 & {\sf 200X} & {\sf ICIC} & {\sf Multi-task learning}\\
   \cline{2-5}
   & 4 & {\sf 200X} & {\sf ICML} & {\sf Statistic models}\\
   \cline{2-5}
  & 5 & {\sf 200X} & {\sf EMNLP} & {\sf Language models}\\
\hline
 \end{tabular}
 \caption{ \label{tab:case}\textbf{Example cells selected by \textit{cube2net} on DBLP.}}
\end{table}

\begin{figure}[h!]
\centering
\subfigure[Increasing Proximity]{
\includegraphics[width=0.24\textwidth]{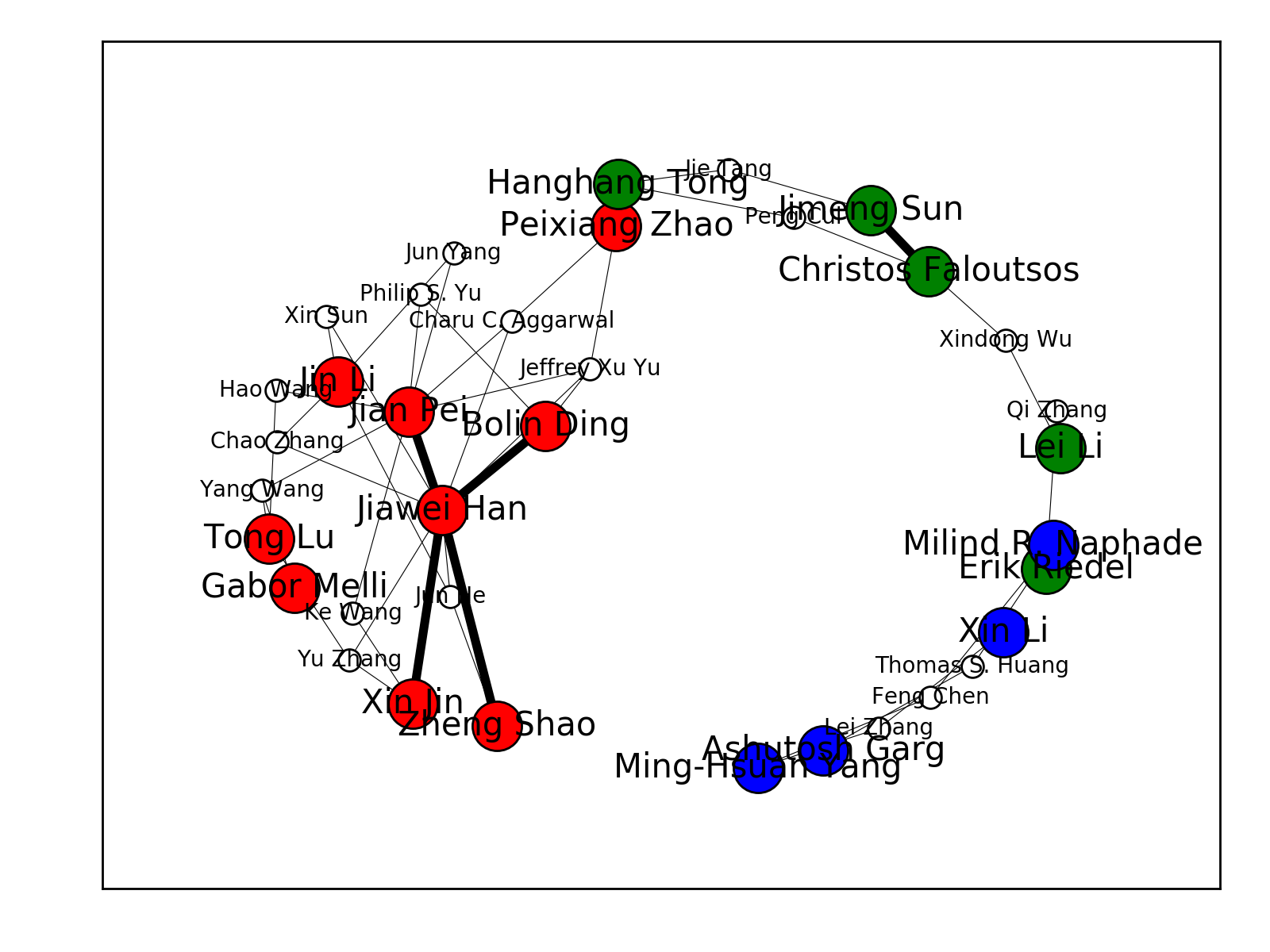}}
\hspace{-12pt}
\subfigure[Decreasing Proximity]{
\includegraphics[width=0.24\textwidth]{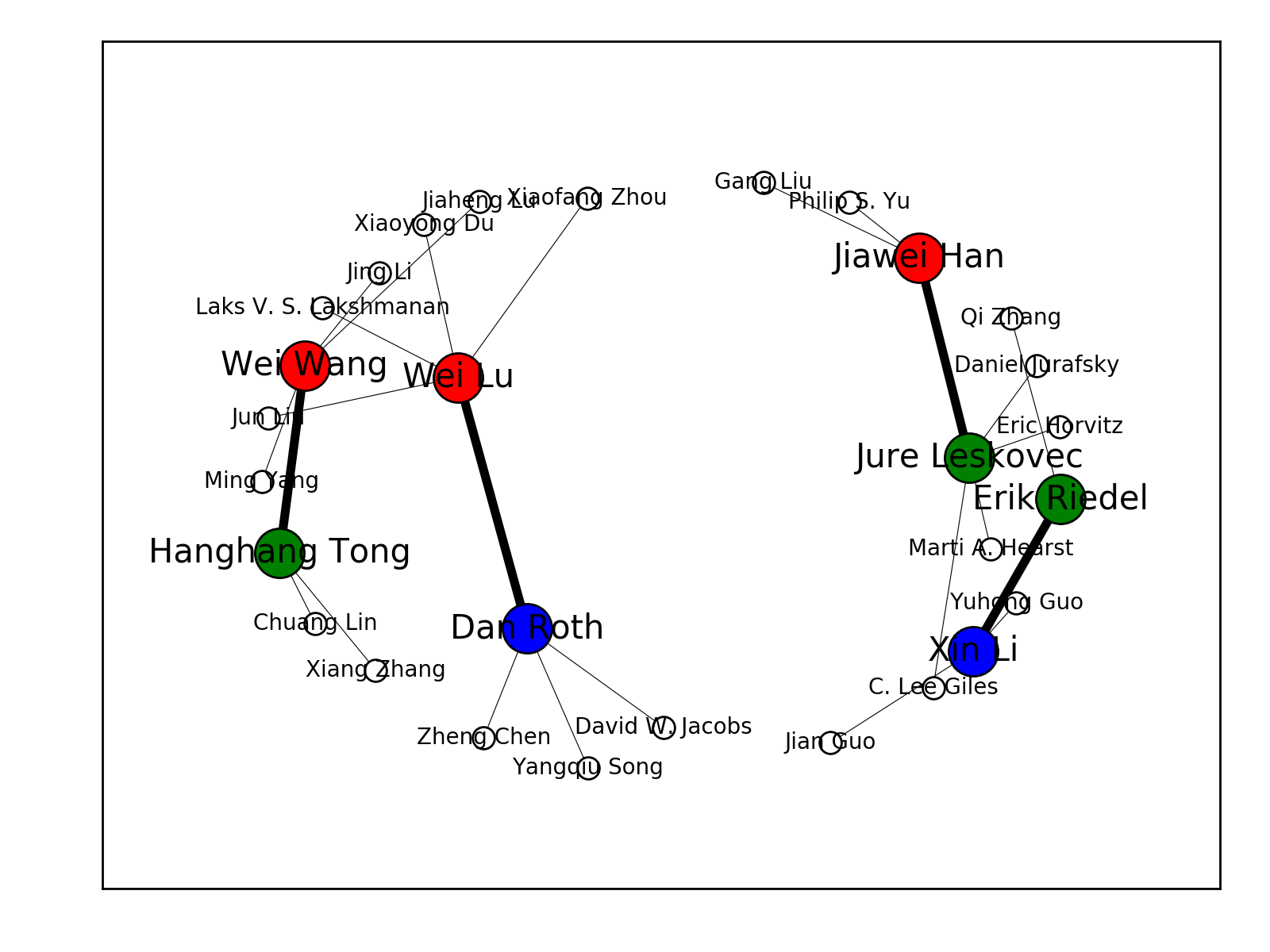}}
\caption{\textbf{Close-ups of the networks constructed by \textit{cube2net} (better viewed with zoom-in).}}
\label{fig:network}
\end{figure}

As a side product of \textit{cube2net}, we are able to gain valuable insight into the query set of objects by looking at the selected cells. 
Table \ref{tab:case} shows the first few cells selected by \textit{cube2net} on DBLP for the small set of 116 authors from four well-known research groups led by authors in the left column. In the table, the rows are intentionally misaligned because there exists no one-to-one correspondence between individual authors and cells.  
However, all cells selected by \textit{cube2net} are indeed highly relevant to the given set of authors, and they together provide a multi-dimensional view of the whole query set. 
By looking at the \textit{decade} dimension, we can see that the majority authors in the set are most active in the 2000s, and \textit{cube2net} does select the most relevant cells in this dimension. However, since the authors have quite different research focuses, the selected labels in the \textit{venue} dimension are quite heterogeneous, which helps drag away authors from different research groups, such as those active in data mining \textit{v.s.}~machine learning. The \textit{topic} dimension further describes the exact research problems that mostly appeal as well as differentiate authors in $\mathcal{Q}$. 

In Figure \ref{fig:network}, we focus on particular parts of the network constructed around the small set of 116 labeled authors to gain insight into how  \textit{cube2net} can construct high-quality networks that facilitate downstream tasks. 
Colored nodes are those in $\mathcal{Q}$, where different colors denote the ground-truth research group labels, and white nodes are those within the cells selected by \textit{cube2net}. The links among nodes in the given set are thicker than those between them and the added nodes.
As we can see, in Figure \ref{fig:network} (a), some authors in the same research groups do not have direct co-authorships, but they are drawn closer by the additional authors who co-author with both of them. Moreover, in Figure \ref{fig:network} (b), some authors in different research groups have co-authorships, but they are drawn further away by their different other co-authors. 
\textit{cube2net} can efficiently find those helpful additional authors, and adding them into the network can effectively improve modeling of the proximities among the query set of authors.

\subsection{Link Prediction On Yelp}
Although we focus on the author networks of DBLP when developing \textit{cube2net}, the concepts and techniques are general and applicable to various real-world networks. To demonstrate this, we conduct experiments on the widely used Yelp dataset. Instead of clustering, we focus on the more popular and well-studied problem of place recommendation. Moreover, since the constructed networks are business-user bipartite networks, instead of the homogeneous network embedding algorithms, we apply standard matrix completion algorithms to evaluate their quality towards link prediction. Superior performance in these experiments can further indicate the general and robust advantages of \textit{cube2net} across different datasets and tasks.

\subsubsection{Experimental Settings}
{\flushleft \bf Dataset.}
We use the public Yelp dataset from the Yelp Challenge Round 11\footnote{https://www.yelp.com/dataset}. Its basic statistics are also in Table \ref{tab:stat}. 

The dataset contains semi-structured data around businesses and users, such as locations, reviews, check-ins, \etc.
To organize the complex multi-dimensional data, we leverage quality attributes of businesses like \textsf{city} and \textsf{category}, and aggregate the \textsf{review texts} of each business to model them similarly as for paper contents in DBLP. Thus, we are able to build a simple data cube with dimensions \textit{city}, \textit{catetory} and \textit{topic}, with each cell holding the corresponding businesses and users. For cell embedding, since all labels are common phrases, we leverage the same pre-trained word embedding table and general cell embedding approach as used for DBLP.

To further demonstrate the power of \textit{cube2net}, we focus on the well studied task of place recommendation on Yelp, which is essentially a link prediction problem on a bipartite network of businesses and users as objects, and reviews as links between them. 
Based on additional attributes of the businesses, we generate two different query sets of objects and hide some of the links among them for performance evaluation. 
Specifically, the first set (ILKID) includes 198 businesses randomly selected from the total 414 that are within the state of \textsf{Illinois} and are \textsf{Good For Kids}, and the second set (NVOUT) includes 2,982 businesses randomly selected from the total 6,020 that are within the state of \textsf{Nevada} and offer \textsf{Takeouts}. Both sets also include 80\% of all users that have left reviews for the corresponding businesses (\ie, 528 users for ILKID and 20,968 users for NVOUT). 20\% of the links generated by the latest reviews are hidden for evaluation (\ie, 161 for ILKID and 6,269 for NVOUT). For each $\mathcal{Q}$ of the two query sets, our task is thus to efficiently bring in relevant businesses and users from the large remaining part of the data, to construct a bipartite network around $\mathcal{Q}$, and improve the performance of link prediction in $\mathcal{Q}$.


{\flushleft \bf Evaluation protocols.}
To evaluate the quality of constructed networks on the particular task of business-user link prediction on Yelp, we chose a popular low-rank matrix completion algorithm called \textit{Soft-Impute} \cite{mazumder2010spectral} with open source implementations\footnote{https://github.com/iskandr/fancyimpute}, due to its relatively high and stable performance and efficiency.
Upon the  bipartite networks constructed by different algorithms in comparison, \textit{Soft-Impute} is applied to recover the missing links, which are then evaluated on the held out links in the query sets.

To evaluate the link prediction performance, we compute precision and recall at $K$, as commonly done in related works \cite{liu2017experimental}. We also record the runtime of network construction and \textit{Soft-Impute} on the same server, as well as the size of the constructed networks.

{\flushleft \bf Parameter settings.}
Since the textual contents in reviews are simpler than in papers, we set the number of \textit{topics} to 50, and we also filter out cells with less than 10 objects. For the reinforcement learning model, we use all the same default parameter settings as in Section \ref{sec:exp}.1, except for the length of trajectories $m$: We use 50 for ILKID and 100 for NVOUT. This also confirms that \textit{cube2net} is not sensitive to parameter settings and can easily maintain robust performance across domains, datasets, and tasks.

\subsubsection{Performance Comparison with Baselines}

\begin{figure*}[h!]
\centering
\vspace{-10pt}
\subfigure[Precision]{
\includegraphics[width=0.245\textwidth]{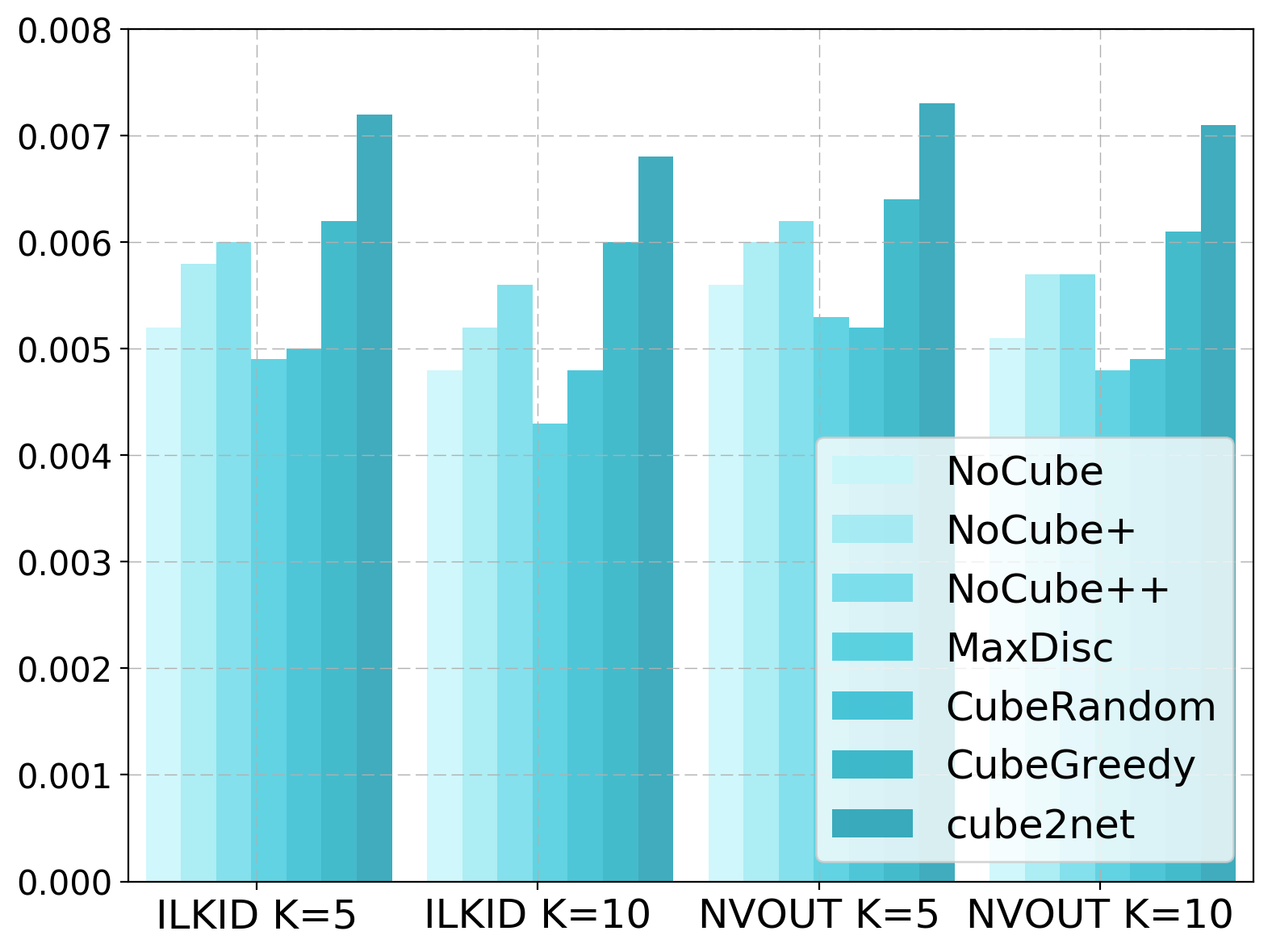}}
\hspace{-8pt}
\subfigure[Recall]{
\includegraphics[width=0.245\textwidth]{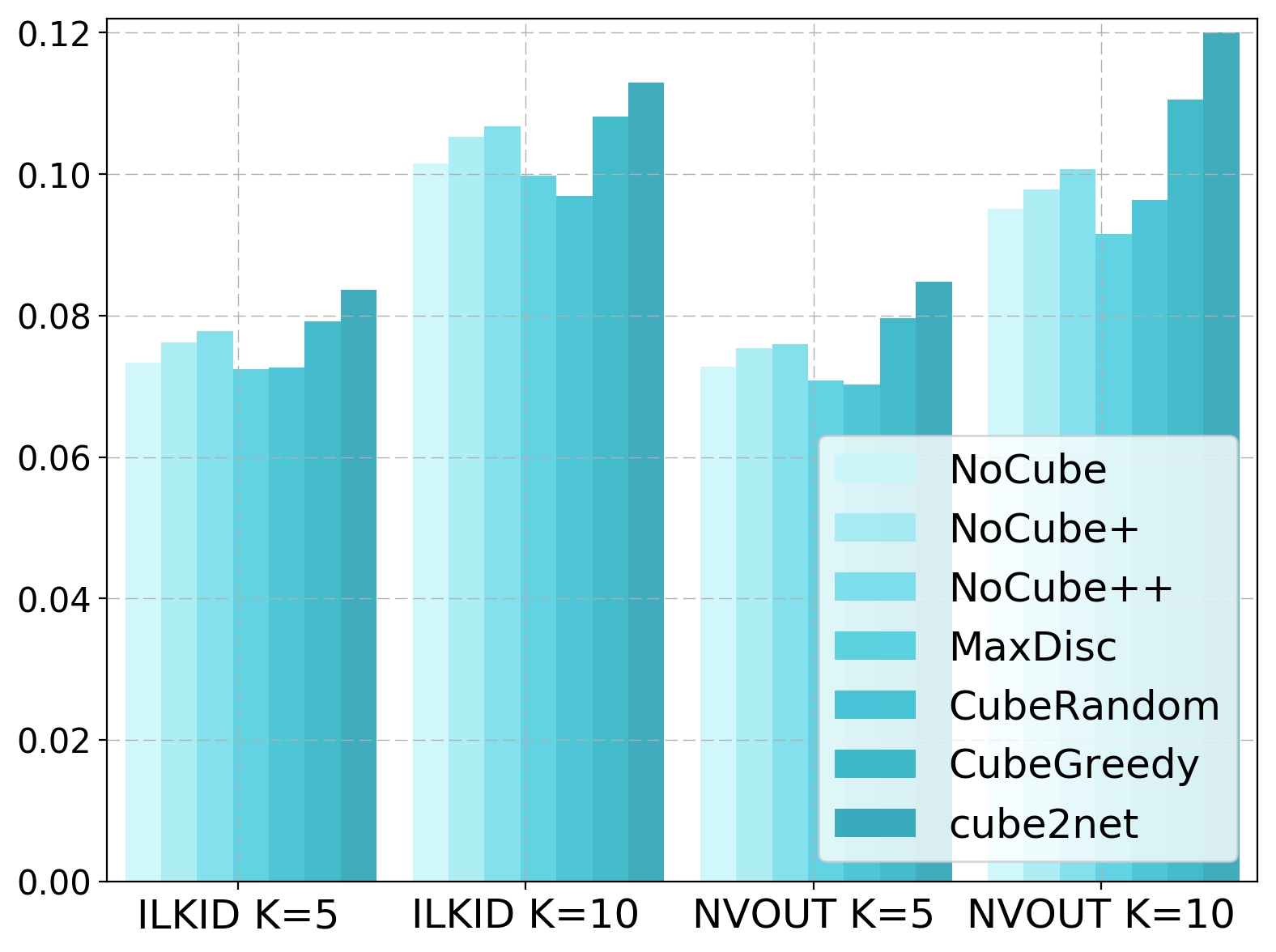}}
\hspace{-8pt}
\subfigure[Runtime / log(s)]{
\includegraphics[width=0.245\textwidth]{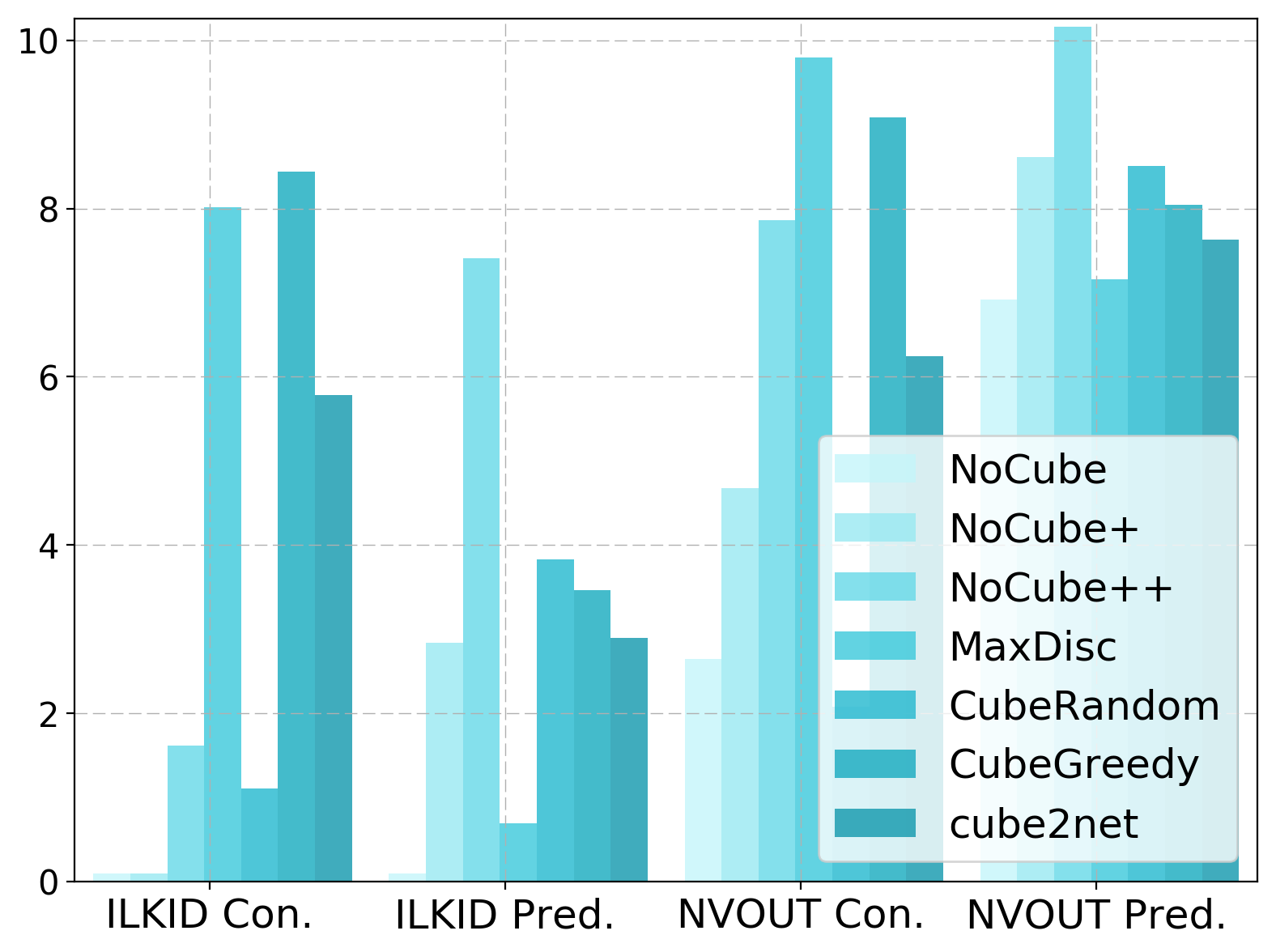}}
\hspace{-8pt}
\subfigure[Network Size / log(\#)]{
\includegraphics[width=0.245\textwidth]{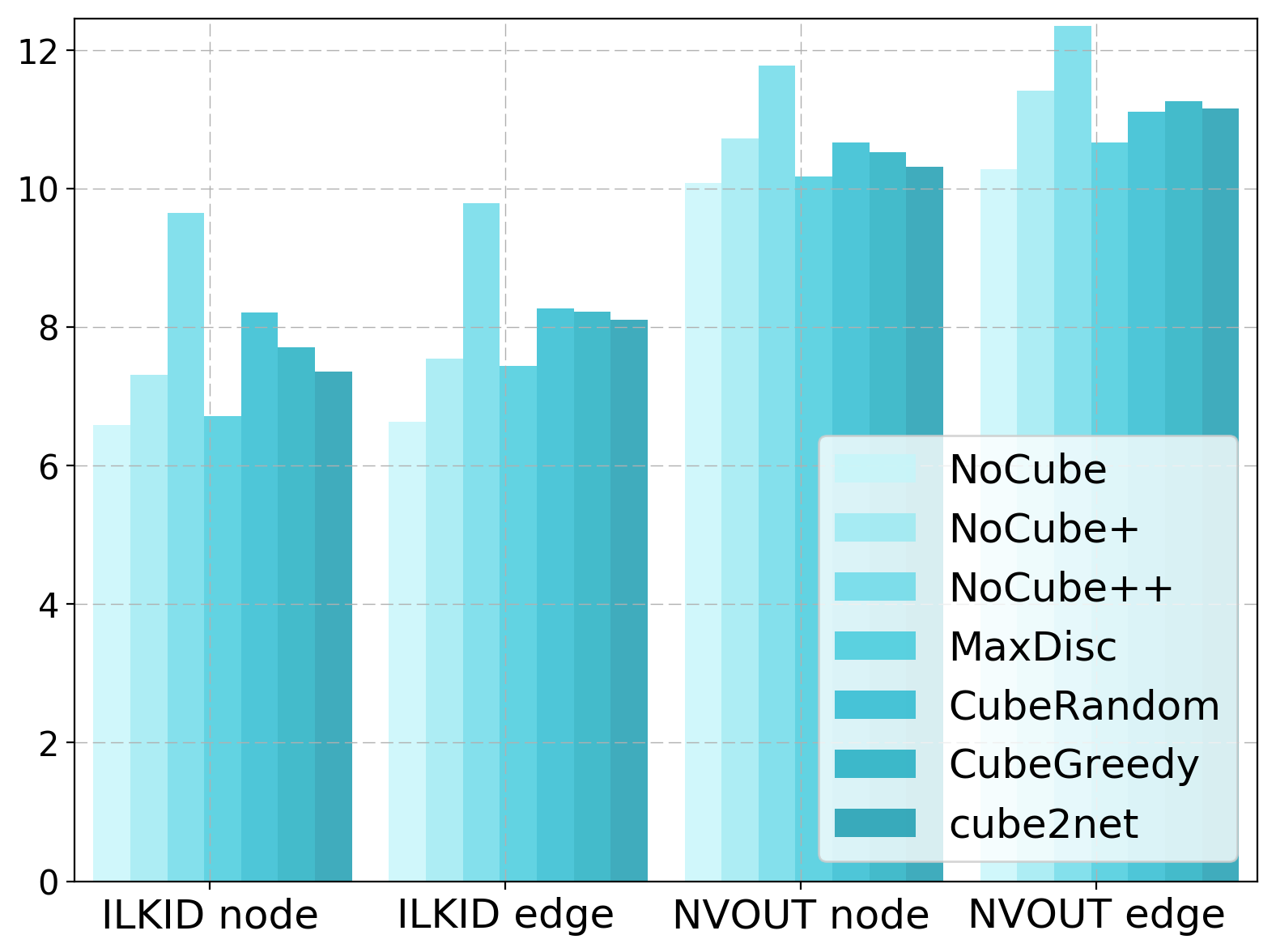}}
\caption{\textbf{Performance of network construction for link prediction on the ILKID Yelp subset.}}
\label{fig:ilkid}
\end{figure*}

Figure \ref{fig:ilkid} demonstrates the quality of constructed query-specific networks on Yelp.

In the figures, similar merits of \textit{cube2net} as discussed on DBLP can be observed, \ie, it is able to construct networks that facilitate downstream tasks with both effectiveness and efficiency.
Particularly, the networks constructed by \textit{cube2net} lead to the best precision and recall for business-user link prediction among all compared algorithms.
Moreover, the networks constructed by \textit{cube2net} have the smallest size, except for \textit{MaxDisc}, which only considers the connected network components.
To compute such a dense and connected network component, \textit{MaxDisc} spends a lot of time searching over the node neighborhoods, but still leaves a lot of nodes dangling outside of the constructed network, resulting in quite poor performance.
While NVOUT (2,982 businesses + 20,968 users) is much larger than ILKID (198 businesses + 161 users), the results follow quite similar trends. Such results strongly indicate the general and robust advantages of \textit{cube2net}.
Finally, for different queries on the corresponding datasets, we use almost the same parameters (except for the trajectory lengths), which confirms the generalization and easy training of \textit{cube2net}.

Note that, compared with DBLP, on Yelp, the node attributes, network properties, downstream tasks and evaluation protocols are all quite different. Therefore, the consistent and significant gains of \textit{cube2net} over all compared algorithms clearly indicate its general power in query-specific network construction to facilitate various network mining algorithms towards different downstream tasks.

\section{Related Work}
\label{sec:related}

\subsection{Network Mining}
Various network mining algorithms have been developed in the past decade, such as relation inference \cite{yang2019relationship, yang2020relation}, structure learning \cite{yang2019query} and structure generation \cite{yang2019conditional}. Recently, unsupervised network embedding based on the advances in neural language models like \cite{mikolov2013distributed} has been extremely popular \cite{perozzi2014deepwalk, tang2015line, grover2016node2vec}. The idea is to compute a distributed representation of nodes, which captures their network proximities, regarding both neighborhood similarity and structural similarity \cite{zhang2017weisfeiler, lyu2017enhancing, ribeiro2017struc2vec, yang2018node, shi2019user}. 

However, almost all existing network mining methods are limited to a fixed set of objects. Among them, most methods assume the network structures to be given \cite{yang2017cone, yang2019place, yang2018similarity, yang2018did, yang2018meta, yang2019relationship}, while some others try to infer links among fixed sets of objects \cite{gomez2010inferring, mcauley2015inferring, hallac2017network}. Some recent methods consider dynamic networks with changing sets of objects \cite{zhou2018dynamic, li2017attributed, ma2018depthlgp}, but they do not actively choose which objects to include into the network. For the particular problem of subnetwork construction, some heuristic search algorithms have been developed to find interesting subnetworks from seed nodes, but they only work for particular network properties and do not consider node attributes \cite{gionis2017bump, sozio2010community}.
Regarding scalability, commercial graph databases like neo4j \cite{vukotic2014neo4j} can handle massive networks efficiently, but they only support fixed sets of operations and do not provide seamless combination to arbitrary graph mining algorithms.

In this work, we consider a more realistic situation where each network mining task is performed on a relatively small query set of objects \cite{yang2019cubenet}, and focus on efficiently leveraging the appropriate portion of the whole massive network to facilitate query-specific knowledge discovery. 
To the best of our knowledge, this is first research effort to tackle network mining through this novel perspective.

%

\subsection{Reinforcement Learning}
Reinforcement learning studies the problem of automated decision making by learning policies to take actions and maximize a reward signal from the environment. Some recent significant progress has been made by combining advances in deep learning for sensory processing with reinforcement learning \cite{krizhevsky2012imagenet}, resulting in the super-human performance on high-dimensional game controls \cite{silver2017mastering}.

As close to our work, several approaches have been proposed to solve the combinatorial optimization problems with reinforcement learning. 
\cite{vinyals2015order} blends an attention mechanism into a sequence-to-sequence structure, in order to learn a conditional probability as a solution to graph mining tasks such as the TSP (Traveling Salesman Problem). They rely on significant training data and can hardly generalize to different tasks or larger networks.
Policy gradient \cite{sutton2000policy} is used to learn the conditional probability policy in \cite{bello2016neural}, while \cite{khalil2017learning} learns a greedy heuristic using deep Q-learning \cite{dai2016discriminative}. In general, they aim to learn policies that generalize to unseen instances from a certain data distribution. These methods have achieved notable successes on some classical graph combinatorial optimization problems. However, when it comes to our scenario of network construction with large action spaces, these algorithms are no longer applicable. 

In this work, we connect continuous control policy gradient with powerful data cube organization through novel cell embedding. To the best of our knowledge, this is the first research effort to leverage reinforcement learning for efficient data cube exploration and quality network construction.

\section{Conclusions}
\label{sec:con}
Network mining has been intensively studied by a wide research community.
In this work, we stress on the novel angle of query-specific network construction, which aims to break the efficiency bottleneck of existing network mining algorithms and facilitate various downstream applications on particular query sets of objects. 
To achieve this goal, we propose to leverage the power of data cube, which significantly benefits the organization of massive real-world networks. 
Upon that, a novel reinforcement learning framework is designed to automatically explore the data cube structures and construct the most relevant query-specific networks.
We demonstrate the effectiveness and efficiency of \textit{cube2net} as a universal network data preprocessor in improving various network mining algorithms for different tasks, with a simple design of data cube and reinforcement learning, while many potential improvements can be explored in future works.

\section*{Acknowledgements}\label{sec:ack}
Research was sponsored in part by U.S. Army Research Lab. under Cooperative Agreement No. W911NF-09-2-0053 (NSCTA), DARPA under Agreement No. W911NF-17-C-0099, National Science Foundation IIS 16-18481, IIS 17-04532, and IIS-17-41317, DTRA HDTRA11810026, and grant 1U54GM114838 awarded by NIGMS through funds provided by the trans-NIH Big Data to Knowledge (BD2K) initiative (www.bd2k.nih.gov).

\subsection*{APPENDIX A: Proof of Theorem \ref{th:rl}}
The $\xi$ neighborhood $\mathcal{N}_\xi(c)$ of cell $c$ is a $\kappa$-dimensional ball around $\mathbf{u}_c$. We firstly write out our actor function $\mathbf{f}(S, c)=\mu(S, c)$ for the current particular state $S$ and cell $c$ being explored. The situation for the critic function $\nu(S,c)$ is exactly the same. According to our neural architecture described in Section \ref{sec:rl}, we have
\begin{align}
\mathbf{f}(S, c) = & \mathbf{g}(\mathbf{S}+\mathbf{u}_c) \nonumber\\
= & \mathbf{h}^Q(\mathbf{h}^{Q-1}(\ldots\mathbf{h}^1(\mathbf{\mathbf{S}+\mathbf{u}_c})\ldots)),
\label{eq:ut}
\end{align}
where
\begin{align}
\mathbf{h}^q(\mathbf{u}) = \text{Sigmoid} (\mathbf{W}^q\mathbf{h}^{q-1}(\mathbf{u})+\mathbf{b}^q).
\label{eq:utf}
\end{align}
$Q$ is the number of layers in the value function, $\mathbf{W}^q$ and $\mathbf{b}^q$ are the parameters of the $q$-th layer. The output of the last layer $\mathbf{h}^Q$ is a $\kappa$-dimensional vector for the actor network $\mu$ (a single number for the critic network $\nu$), and $\mathbf{h}^0(\mathbf{u})=\mathbf{u}$.

Now we prove the first property, \ie, $\forall c'\in\mathcal{N}_\xi(c), ||\mathbf{f}(S, c)-\mathbf{f}(S,c')||^2_2\leq \eta\xi$.
\begin{proof}
Since $c'\in\mathcal{N}_\xi(c)$, we have
\begin{align}
||(\mathbf{S} + \mathbf{u}_c) - (\mathbf{S} + \mathbf{u_{c'}})||^2_2 = || \mathbf{u}_{c'} - \mathbf{u}_c||_2^2 \leq \xi.
\end{align}
According to Eq.~\ref{eq:ut} and \ref{eq:utf}, $\mathbf{g}(\cdot)$ is \textit{smooth}, \ie, $\mathbf{g} \in C^{\infty}$ (this conclusion is actually non-trivial, but its proof is beyond the scope of this work). Based on the smoothness of $\mathbf{g}$, we have
\begin{align}
 \forall \xi, \mathbf{S}, \mathbf{u}_c, \mathbf{u}_{c'}, & ||\mathbf{g}(\mathbf{S}+\mathbf{u}_c)-\mathbf{g}(\mathbf{S}+\mathbf{u}_{c'})||_2^2  \nonumber\\
\leq & ||\nabla \mathbf{g}||_2^2 || \mathbf{u}_{c'} - \mathbf{u}_c||_2^2,
\label{eq:1a}
\end{align}
where
\begin{align}
||\nabla \mathbf{g}||_2^2 & = \text{sup}_{\mathbf{u}\in \Omega} ||\nabla \mathbf{g}(\mathbf{u})||_2^2, \nonumber\\
\Omega & = \{\mathbf{S} + \lambda \mathbf{u}_c + (1-\lambda) \mathbf{u}_{c'}\}_{\lambda\in[0,1]}.
\label{eq:1b}
\end{align}
Therefore, we have
\begin{align}
\forall \xi, S, c, & c'\in \mathcal{N}_\xi(c), \exists \eta < \infty, \text{\it s.t.} \nonumber\\
&|\mathbf{f}(S, c)-\mathbf{f}(S,c')| \leq \eta\xi\qedhere.
\label{eq:1c}
\end{align}
\end{proof}

Now we prove the second property, \ie, $\forall c'\in\mathcal{N}_\xi(c), \\|| \nabla \mathbf{f}(S,c) - \nabla \mathbf{f}(S,c') ||^2_2 \leq  \zeta\xi$. 
\begin{proof}
Similarly as in Eq.~\ref{eq:1a}, based on the smoothness of $\mathbf{g}$ (therefore, the smoothness of $\nabla \mathbf{g}$), we have
\begin{align}
 \forall \xi, \mathbf{S}, \mathbf{u}_c, \mathbf{u}_{c'}, & ||\nabla \mathbf{g}(\mathbf{S}+\mathbf{u}_c)-\nabla \mathbf{g}(\mathbf{S}+\mathbf{u}_{c'})||^2  \nonumber\\
\leq & ||\mathbf{H}||_2^2 || \mathbf{u}_{c'} - \mathbf{u}_c||_2^2,
\end{align}
where $\mathbf{H}$ is the Hessian matrix of $\mathbf{g}$, and
\begin{align}
||\mathbf{H}||_2^2 & = \text{sup}_{\mathbf{u}\in \Omega} ||\mathbf{H}(\mathbf{u})||_2^2, \nonumber\\
\Omega & = \{\mathbf{S} + \lambda \mathbf{u}_c + (1-\lambda) \mathbf{u}_{c'}\}_{\lambda\in[0,1]}.
\end{align}
Therefore, similarly to $\mathbf{f}$, we have
\begin{align}
\forall \xi, S, c, & c'\in \mathcal{N}_\xi(c), \exists \zeta < \infty, \text{\it s.t.} \nonumber\\
&|| \nabla \mathbf{f}(S,c) - \nabla \mathbf{f}(S,c') ||^2_2 \leq  \zeta\xi\qedhere.
\end{align}
\end{proof}

\subsection*{APPENDIX B: Details of Cell Embedding}
To enable efficient Gaussian exploration in a continuous space to fully leverage the data cube structure, we propose and design the novel process of \textit{cell embedding} upon existing data cube technology. Particularly, we aim to compute a distributional representation for each cell, which captures their semantics \wrt~the corresponding multi-dimensional labels. As we show in Section \ref{sec:rl}, since such continuous cell representations can facilitate the utility estimation of close-by cells efficiently by avoiding the need of exhaustive search, they are critical to the success of an efficient reinforcement learning algorithm for efficient cube exploration and network construction.
%
%

We introduce a straightforward way of computing the cell embedding, based on the recent success of word embedding techniques, and theoretically show the effectiveness of it in capturing the semantics of cells. Following is a summarization of our cell embedding approach:
\begin{itemize}[leftmargin=20pt]
\item \textbf{Step 1}: For each cell $c\in\mathcal{C}$, we firstly decompose it into the corresponding set of labels in the multiple cube dimensions, \ie, $c = \{l_c^1, l_c^2, \ldots, l_c^P\}$, where each $l_c^p\in \mathcal{L}^p$ is the label of $c$ in the $p$-th dimension. The embedding of $c$ is then computed as $\mathbf{u}_c = [\mathbf{u}_{l_c^1}, \mathbf{u}_{l_c^2}, \ldots, \mathbf{u}_{l_c^P}]$, where $\mathbf{u}_{l_c^p}$ is the label embedding of label $l_c^p$ and $[\ldots]$ is the vector concatenation operator.
\item \textbf{Step 2}: For each label $l\in\mathcal{L}$ ($\mathcal{L}=\cup_{p=1}^P\mathcal{L}^p$), we firstly split it into single words, \ie, $l=\{w_l^1, w_l^2, \ldots\}$. The embedding of $l$ is then computed as $\mathbf{u}_l=\sum_i \mathbf{u}_{w_l^i}$.
\item \textbf{Step 3}: For each word $w\in\mathcal{W}$, where $\mathcal{W}$ is an assumably complete vocabulary, look up the embedding of $w$, \ie, $\mathbf{u}_w$, from a pre-trained word embedding table. For generality, we use the popular GloVe table\footnote{https://nlp.stanford.edu/projects/glove/}. Additionally, we assign the zero vector to all stop words and unmatched words.
\end{itemize}

Let us still take DBLP data as an example. Consider a cell $c$, \eg, $c=<$\textsf{200X, SIGMOD, data mining}$>$, which has three dimensions, \ie, \textit{decade}, \textit{venue} and \textit{topic}. 
For the \textit{decade} dimension, the label \textsf{200X} is actually an aggregation of 10 words (2000-2009). Therefore, we can directly look up the 10 words from the embedding table and add them up as the label representation of \textsf{200X}. 
For the texual dimensions like \textit{topic}, we firstly split the phrase \textsf{data mining} into two words, \ie, \textsf{data} and \textsf{mining}, then look up their embeddings in the table, and finally add them up. The \textit{venue} label \textsf{SIGMOD} is firstly mapped back to its full name \textsf{Sig on management of data} and then processed in the same way as \textit{topic}. After getting the embeddings of all three labels, the cell embedding is computed as a $\kappa$-dimensional vector concatenation of the corresponding label embeddings.

Although this approach of generating cell embedding is straightforward, we theoretically show that in this way, we can actually capture the semantic proximities among cells. Specifically, we prove the following theorem.

\begin{theorem}
Given our particular cell embedding approach, $\forall c_i, c_j \in \mathcal{C}, ||\mathbf{u}_{c_i} - \mathbf{u}_{c_j}||_2^2 \leq \epsilon S(c_i, c_j)$, where $\epsilon$ is a constant and $S(c_i, c_j)$ is the semantic gap between $c_i$, $c_j$.
\label{th:cube}
\end{theorem}

\begin{proof}
This proof is based on the assumption that the pre-trained word embedding $\mathcal{U}_W$ has captured word semantics in the vocabulary $\mathcal{W}$, so that $\forall w_i, w_j \in \mathcal{W}, ||\mathbf{u}_{w_i}-\mathbf{u}_{w_j}||^2_2<\epsilon S(w_i, w_j)$, where $S(w_i, w_j)$ is the semantic gap between $w_i$ and $w_j$. 

We further define the semantic gaps among labels. 
Specifically, $\forall l_i, l_j \in \mathcal{L}^p$, where $p\in\mathcal{P}$ is any cube dimension, suppose $l_i = \{w_i^1, w_i^2, \ldots, w_i^{K_i}\}$, where $K_i$ is the number of words in $l_i$; $l_j = \{w_j^1, w_j^2, \ldots, w_j^{K_j}\}$, where $K_j$ is the number of words in $l_j$, the semantic gap $S(l_i, l_j)$ can then be defined as $S(l_i, l_j) = \sum_{k_i=1}^{K_i}\sum_{k_j=1}^{K_j} S(w_i^{k_i}, w_j^{k_j})$.

Now we first prove that the label embedding $\mathcal{U}_L$ captures the label semantics, so that $\forall l_i, l_j \in \mathcal{L}^p, ||\mathbf{u}_{l_i} - \mathbf{u}_{l_j}||_2^2 \leq \epsilon S(l_i, l_j)$. 
For the label embeddings $\mathcal{U}_L$, ($\forall \mathbf{u}_l \in \mathcal{U}_L, \mathbf{u}_l=\sum_{k=1}^{K} \mathbf{u}_{w_l^k}$),
\begin{align}
&\sum_{k_i=1}^{K_i} \sum_{k_j=1}^{K_j} || \mathbf{u}_{w_i^{k_i}} - \mathbf{u}_{w_j^{k_j}} ||^2_2 - ||\mathbf{u}_{l_i}-\mathbf{u}_{l_j}||^2_2 \nonumber\\
= & (K_j - 1) \sum_{k_i=1}^{K_i} \mathbf{u}_{w_i^{k_i}}^T  \mathbf{u}_{w_i^{k_i}} + (K_i - 1) \sum_{k_j=1}^{K_j} \mathbf{u}_{w_j^{k_j}}^T  \mathbf{u}_{w_j^{k_j}} \nonumber\\
&- 2 \sum_{k_i^1 \neq k_i^2; k_i^1, k_i^2=1}^{K_i} \mathbf{u}_{w_i^{k_i^1}}^T \mathbf{u}_{w_i^{k_i^2}} - 2 \sum_{k_j^1 \neq k_j^2; k_j^1, k_j^2=1}^{K_j} \mathbf{u}_{w_j^{k_j^1}}^T \mathbf{u}_{w_j^{k_j^2}} \nonumber\\
= & (K_j - K_i)(\sum_{k_i=1}^{K_i} \mathbf{u}_{w_i^{k_i}}^T  \mathbf{u}_{w_i^{k_i}} - \sum_{k_j=1}^{K_j} \mathbf{u}_{w_j^{k_j}}^T  \mathbf{u}_{w_j^{k_j}}) \nonumber\\
& + \sum_{k_i^1 \neq k_i^2; k_i^1, k_i^2=1}^{K_i} || \mathbf{u}_{w_i^{k_i^1}} - \mathbf{u}_{w_i^{k_i^2}} ||^2_2 \nonumber\\
& + \sum_{k_j^1 \neq k_j^2; k_j^1, k_j^2=1}^{K_j} || \mathbf{u}_{w_j^{k_j^1}} - \mathbf{u}_{w_j^{k_j^2}} ||^2_2 
\geq  0.
\end{align}
Therefore, we have
\begin{align}
&||\mathbf{u}_{l_i}-\mathbf{u}_{l_j}||^2_2 \leq \sum_{k_i=1}^{K_i} \sum_{k_j=1}^{K_j} || \mathbf{u}_{w_i^{k_i^1}} - \mathbf{u}_{w_j^{k_j^1}} ||^2_2 \nonumber\\
\leq & \epsilon \sum_{k_i=1}^{K_i}\sum_{k_j=1}^{K_j} S(w_i^{k_i}, w_j^{k_j}) = \epsilon S(l_i, l_j) 
\end{align}

Moreover, $\forall c_i, c_j \in \mathcal{C}$, suppose $c_i = \{l_i^1, l_i^2, \ldots, l_i^{P}\}$, and $c_j = \{l_j^1, l_j^2, \ldots, l_j^{P}\}$, where $P$ is the number of cube dimensions, we then define the semantic gap $S(c_i, c_j)$ as
\begin{align}
S(c_i, c_j) = \sum_{p=1}^{P} S(l_i^p, l_j^p).
\end{align} 

Now we prove that the cell embedding $\mathcal{U}_C$ captures the cell semantics, so that $\forall c_i, c_j \in C, ||\mathbf{u}_{c_i} - \mathbf{u}_{c_j}||_2^2 \leq \epsilon S(c_i, c_j)$.

For the cell embeddings $\mathcal{U}_C$, ($\forall \mathbf{u}_{c} \in \mathcal{U}_C, \mathbf{u}_{c}=[\mathbf{u}_{l_{c}^1}, \ldots, \mathbf{u}_{l_{c}^P}]$, where $[\ldots]$ denotes vector concatenation), the proof is trivial because the norm of a concatenated vector equals to the sum of the norms of its components, \ie,
\begin{align}
& ||\mathbf{u}_{c_i} - \mathbf{u}_{c_j}||^2_2 \nonumber\\
= & \sum_{p=1}^P ||\mathbf{u}_{l_i^{p}} - \mathbf{u}_{l_j^{p}}||^2_2 \nonumber\\
\leq & \epsilon \sum_{p=1}^P S(l_i^p, l_j^p) = \epsilon S(c_i, c_j)
\qedhere.
\end{align}
\end{proof}

\bibliographystyle{IEEEtran}
\small
\bibliography{carlyang} 
\end{document}